\documentclass[aps,pra,preprint,titlepage,superscriptaddress,longbibliography]{revtex4-2}
\usepackage{amsmath}
\usepackage{amsthm}
\usepackage{amsfonts}
\usepackage{amssymb}
\usepackage{xcolor,graphicx}
\usepackage[pdftex,hidelinks]{hyperref}
\usepackage{bm}

\usepackage{bibunits}
\defaultbibliographystyle{apsrev4-2}

\usepackage{etoolbox}
\makeatletter
\patchcmd{\NAT@bibsetnum}{\ref{LastBibItem}}{\ref{LastBibItem-\@bibunitname}}{}{
	\PackageError{manuscript}{Cannot patch bibliography label reader}{}
}
\patchcmd{\endthebibliography}{\label{LastBibItem}}{\label{LastBibItem-\@bibunitname}}{}{
	\PackageError{manuscript}{Cannot patch bibliography label writer}{}
}
\renewcommand*{\@extra@b@citeb}{.\@bibunitname}
\renewcommand*{\@extra@binfo}{.\@bibunitname}
\g@addto@macro\frontmatter@setup{\linespread{1}\selectfont}
\g@addto@macro\frontmatter@abstractfont{\linespread{1}\selectfont}
\makeatother
\newtheorem{theorem}{Theorem}
\newtheorem{proposition}[theorem]{Proposition}
\newtheorem{lemma}[theorem]{Lemma}
\newtheorem{corollary}[theorem]{Corollary}
\newtheorem{definition}[theorem]{Definition}
\newtheorem{remark}[theorem]{Remark}

\begin{document}
	\raggedbottom
	
	\begin{bibunit}[apsrev4-2]

		\title{Root-system structure of sloppiness in passive Gaussian metrology}
		
		\author{A. Kafuri} %orcid: 0009-0009-5865-2611
		\email[e-mail: ]{anuar.kafuri@cinvestav.mx}
		\affiliation{Departamento de F\'isica, Centro de Investigaci\'on y de Estudios Avanzados del IPN, P.O. Box 14-740, CP. 07000, Ciudad de M\'exico, Mexico.}
		
		\author{B.~M. Rodr\'iguez-Lara}% ORCID:0000-0002-7014-0450
		\email[e-mail: ]{blas.rodriguez@gmail.com}
		\affiliation{Universidad Polit\'ecnica Metropolitana de Hidalgo, Tolcayuca, Hidalgo 43860, Mexico.}
		
		\date{\today}
		
		\begin{abstract}
			Passive transformations of an $n$-mode squeezed probe become locally unidentifiable when the quantum Fisher matrix loses rank.
			For pure zero-mean Gaussian probes, we show that the $C_{n}$ restricted-root system of $\mathrm{Sp}(2n,\mathbb{R})/\mathrm{U}(n)$ governs this exact sloppiness.
			For squeezing magnitudes $r_{1},\ldots,r_{n}$ in the canonical basis, the roots $2r_{j}$ and $r_{j}\pm r_{k}$ label the local phase rotations and two beam-splitter quadratures, and their vanishing identifies every Fisher-null direction.
			The associated Fisher weights quantify the approach to each singular wall.
			At fixed positive mean photon number, we solve the E-optimal probe-design problem of maximizing the smallest passive Fisher eigenvalue in two explicit generator normalizations.
			In the canonical phase and beam-splitter angle convention, the unique optimum in the fundamental Weyl chamber is an arithmetic progression approaching the consecutive-odd-integer dual-Weyl direction at large resource.
			With an invariant generator norm, the smallest eigenvalue is independent of the passive frame and the optimum lies exactly on that ray at every resource.
			Uniform squeezing instead minimizes the local-phase A-optimal cost while leaving one beam-splitter quadrature unidentifiable for every mode pair.
			Finally, the two quadratures of each identifiable mode pair saturate the quantum-geometric incompatibility bound, and the largest compatible passive submodel has $n(n+1)/2$ parameters at a regular spectrum.
			The resulting classification separates local identifiability, sensitivity optimization, and simultaneous attainability.
		\end{abstract}
		
		\maketitle

		%%%%%%%%%%%%%%%%%%%%%%%    Body    %%%%%%%%%%%%%%%%%%%%%%%%%
		
		%%%%%%%%%%%%%%%%%%%%%%%%%%
		\section{Introduction}
		\label{sec:Sec1}
		%%%%%%%%%%%%%%%%%%%%%%%%%%
		
		Gaussian probes provide a standard platform for multiparameter quantum metrology.
		They are completely characterized in phase space by their first moments and covariance matrices, which transform through finite-dimensional matrix operations under experimentally realizable Gaussian processes~\cite{Simon1988p3028,Weedbrook2012p621,Serafini2017}.
		Gaussian-state methods give explicit expressions for the quantum Fisher information (QFI)~\cite{Monras2013p1303.3682,Safranek2019p035304}, while variational approaches optimize Gaussian resources for a broad class of sensing problems~\cite{Safranek2016p062313,Matsubara2019p033014,Gessner2020p3817}.
		These tools make individual estimation problems tractable, but they do not by themselves expose a global structure of parameter identifiability.
		
		We use exact sloppiness for Fisher-rank loss, which renders some local parameter directions unidentifiable and leaves only particular combinations of the encoded parameters estimable~\cite{Frigerio2026p2540001}.
		This exact rank deficiency is the singular limit of the broader use of sloppiness for a wide hierarchy of Fisher eigenvalues in classical models, which need not have exact null directions~\cite{Gutenkunst2007pe189,Transtrum2015p010901}.
		In a Gaussian model with equally squeezed inputs and relative displacement allowed, the state encoded by two successive phase shifts on the same interferometer arm depends only on their sum when no operation separates them.
		An inserted squeezer and suitable interferometer settings can lift this redundancy and yield a compatible model~\cite{Frigerio2026p2540001}.
		Related studies control sloppiness through qubit scrambling~\cite{He2025p325301}, an intermediate phase shift between squeezing operations~\cite{Sharma2025p104511}, and tunable weak measurements in two-phase estimation~\cite{Bizzarri2025p432}.
		Even with full Fisher rank, noncommuting optimal measurements can prevent simultaneous saturation of the individual symmetric-logarithmic-derivative (SLD) bounds~\cite{Matsumoto2002p3111,Ragy2016p052108,DemkowiczDobrzanski2020p363001,Chang2026p126}.
		
		Recent work has begun to organize identifiability, sensitivity, and compatibility geometrically and algebraically.
		The Gaussian QFI admits symplectic decompositions connected with the Siegel metric of the pure Gaussian-state manifold~\cite{Chatterjee2026p045005}.
		Cartan structure can organize precision, sloppiness, and compatibility in gate metrology~\cite{Fazio2025p100260}, while the full symplectic symmetry $\mathrm{Sp}(2n,\mathbb{R})$ provides a framework for multimode Gaussian interferometry and control~\cite{Lv2026p2606.25768}.
		Related mode-estimation approaches show how structural decompositions can convert general precision bounds into explicit probe-design rules~\cite{Gessner2023p996}.
		In distributed Gaussian sensing, the sensitivity matrix of a fixed interferometer network selects the best-estimable linear combination of phases~\cite{Malitesta2023p032621}.
		Its eigenvalue spectrum describes phase sensitivity, whereas the squeezing spectrum in our analysis characterizes the probe.
		
		Here we show that, in its canonical squeezing basis, a product of independent single-mode squeezed vacua carries a generator-resolved identifiability structure for passive Gaussian transformations, which conserve total excitation number and are generated by local phase rotations and beam-splitter couplings.
		In Section~\ref{sec:Sec2}, we use the Bloch--Messiah decomposition~\cite{Bloch1962p95,Braunstein2005p513,Weedbrook2012p621,Serafini2017} to reduce any pure zero-mean Gaussian probe to this product through a passive change of mode basis and show how the squeezing spectrum determines the passive Fisher weights and, through the correspondingly conjugated generators, which parameter combinations remain estimable and where the Fisher matrix loses rank.
		We organize the probe's squeezing magnitudes $r_{1},\ldots,r_{n}$ in the $C_{n}$ Weyl chamber.
		Using established Gaussian Fisher formulas~\cite{Sparaciari2016p023810,Safranek2016p062313}, symmetric-space radial factors~\cite{Helgason1984,Helgason2001}, and the pure-Gaussian metric~\cite{Chatterjee2026p045005}, we identify the singular Fisher set with the $C_{n}$ restricted-root arrangement of $\mathrm{Sp}(2n,\mathbb{R})/\mathrm{U}(n)$ and associate every root hyperplane with the phase rotation or beam-splitter quadrature that becomes locally unidentifiable in Sec.~\ref{sec:Sec3}, providing a physical classification of passive calibration blind spots.
		We convert this singular arrangement into probe-design criteria in Sec.~\ref{sec:Sec4}, derive the dual-Weyl solution of the algebraic root-margin problem, and solve the finite-photon-number maximin problem exactly in both generator normalizations using the E-optimal criterion of classical experimental design~\cite{Kiefer1974p849,Pukelsheim1993}.
		In the canonical angle normalization and fundamental Weyl chamber, we obtain a unique arithmetic progression of squeezing magnitudes.
		We use this classification in Sec.~\ref{sec:Sec5} to compare local-phase sensing with passive-network calibration and show that uniform squeezing minimizes the local-phase A-optimal cost while making one beam-splitter quadrature locally unidentifiable for every mode pair.
		Finally, we derive the complete passive mean-commutator tensor in Sec.~\ref{sec:Sec6}, show that every identifiable beam-splitter pair saturates the quantum-geometric incompatibility bound, and obtain a complete compatibility classification with a sharp bound on compatible-submodel dimension.
		We distinguish identifiability, sensitivity design, and attainability, and close with a summary and our conclusions in Sec.~\ref{sec:Sec7}.
		
		%%%%%%%%%%%%%%%%%%%%%%%%%%%%%%%%%%%%%%%%%%%%%%%%%%%%%%%%%%%%%%%%%%%%%%%%%%%%%%
		\section{Canonical squeezed probe}
		\label{sec:Sec2}
		%%%%%%%%%%%%%%%%%%%%%%%%%%%%%%%%%%%%%%%%%%%%%%%%%%%%%%%%%%%%%%%%%%%%%%%%%%%%%%
		
		We consider $n$ bosonic modes with annihilation operators $\hat{a}_{j}$, $j=1,\ldots,n$, and the canonical zero-mean pure Gaussian probe,
		\begin{align}
			\lvert \psi_{\bm{r}} \rangle =&~ \hat{S}(\bm{r}) \lvert 0 \rangle,
			\qquad
			\hat{S}(\bm{r}) = \prod_{j=1}^{n} e^{\frac{1}{2} r_{j} \left( \hat{a}_{j}^{2} - \hat{a}_{j}^{\dagger 2} \right)},
		\end{align}
		with multimode squeezing unitary $\hat{S}(\bm{r})$, multimode vacuum $\lvert 0 \rangle$, and real squeezing spectrum $\bm{r} = \left( r_{1}, \ldots, r_{n} \right) \in \mathbb{R}^{n}$.
		We refer to unitary transformations that conserve the total excitation number as passive.
		We focus on passive identifiability in the diagonal squeezing basis because any pure zero-mean Gaussian state in Bloch--Messiah decomposition form~\cite{Bloch1962p95,Braunstein2005p513,Weedbrook2012p621,Serafini2017},
		\begin{align}
			\lvert \Psi_{\bm{r}} \rangle =&~ \hat{U}_{1} \hat{S}(\bm{r}) \hat{U}_{2} \lvert 0 \rangle
			= \hat{U}_{1} \lvert \psi_{\bm{r}} \rangle,
		\end{align}
		reduces to the canonical zero-mean pure Gaussian probe under the passive unitary change of reference frame $\hat{U}_{1}^{\dagger}$, while $\hat{U}_{2}$ leaves the multimode vacuum invariant.
		Two probes with the same squeezing spectrum $\bm{r}$ need not share the same phase-rotation and beam-splitter operators in the laboratory frame.
		
		We use an $n^{2}$-element Hermitian basis for the passive algebra $\mathrm{u}(n)$,
		\begin{align}
			\begin{aligned}
				&\hat{g}^{\mathrm{rot}}_{j} = \hat{a}_{j}^{\dagger} \hat{a}_{j}, \qquad &&
				j = 1,\ldots,n, \\
				&\begin{aligned}
					\hat{g}^{\mathrm{bsX}}_{jk} =&~ \hat{a}_{j}^{\dagger} \hat{a}_{k} + \hat{a}_{k}^{\dagger} \hat{a}_{j}, \\
					\hat{g}^{\mathrm{bsY}}_{jk} =&~ -i \left( \hat{a}_{j}^{\dagger} \hat{a}_{k} - \hat{a}_{k}^{\dagger} \hat{a}_{j} \right),
				\end{aligned} \qquad &&
				1 \leq j < k \leq n,
			\end{aligned}
		\end{align}
		one local phase rotation per mode and two beam-splitter quadratures per mode pair.
		This convention sets the relative scales of the parameters, and we use the Euclidean norm of these coefficients to define the spectral design criteria.
		Rescaling the generators can change the optimal probe, although any invertible change of coordinates preserves Fisher rank.
		A coefficient is a local phase angle or the angle in a two-mode coupling whose transmissivity is $\cos^{2}\theta$.
		Each passive generator has the number-conserving quadratic form $\hat{g}_{\mu}=\hat{\bm{a}}^{\dagger}\bm{h}_{\mu}\hat{\bm{a}}$, where $\hat{\bm{a}}=(\hat{a}_{1},\ldots,\hat{a}_{n})^{\mathrm{T}}$ is the annihilation-operator column and $\bm{h}_{\mu}$ is an $n\times n$ Hermitian generator matrix~\cite{Arvind1995p471,Weedbrook2012p621}.
		The invariant inner product $\operatorname{Re}\operatorname{Tr}(\bm{h}_{\mu}\bm{h}_{\nu})$, with the trace over the $n$-dimensional mode space, gives squared norm one to each rotation and two to each beam-splitter quadrature.
		This angle basis is orthogonal but not orthonormal in the invariant metric.
		Using the ordered basis,
		\begin{align}
			\left\{ \hat{g}_{\mu} \right\}_{\mu=1}^{n^{2}} =&~ \left\{ \hat{g}^{\mathrm{rot}}_{j} \right\}_{j=1}^{n} \cup \left\{ \hat{g}^{\mathrm{bsX}}_{jk}, \hat{g}^{\mathrm{bsY}}_{jk} \right\}_{1\leq j<k\leq n},
		\end{align}
		we encode the real parameter vector $\bm{\theta}=(\theta_{1},\ldots,\theta_{n^{2}})\in\mathbb{R}^{n^{2}}$ in the probe as
		\begin{align}
			\lvert \psi_{\bm{r}}(\bm{\theta}) \rangle =&~  e^{-i \sum_{\mu=1}^{n^{2}} \theta_{\mu} \hat{g}_{\mu}} \lvert \psi_{\bm{r}} \rangle,
		\end{align}
		and consider local estimation around the identity, $\bm{\theta}=\bm{0}$, where the generator fluctuations on the canonical squeezed probe determine the quantum Fisher matrix,
		\begin{align}
			\begin{aligned}
				F_{\mu\nu} =&~ 4 \, \mathrm{Re} \left[ \langle \partial_{\mu} \psi \vert \partial_{\nu} \psi \rangle - \langle \partial_{\mu} \psi \vert \psi \rangle \langle \psi \vert \partial_{\nu} \psi \rangle \right], \\
				=&~ 2 \left\langle \left\{ \Delta \hat{g}_{\mu}, \Delta \hat{g}_{\nu} \right\} \right\rangle,
				\qquad
				\Delta \hat{g}_{\mu} = \hat{g}_{\mu} - \left\langle \hat{g}_{\mu} \right\rangle,
			\end{aligned}
		\end{align}
		where $\lvert\psi\rangle=\lvert\psi_{\bm{r}}(\bm{\theta})\rangle$ and $\partial_{\mu}=\partial/\partial\theta_{\mu}$ is the derivative with respect to the parameter $\theta_{\mu}$.
		We evaluate the derivatives at $\bm{\theta}=\bm{0}$ and the expectation values in $\lvert\psi_{\bm{r}}\rangle$.
		Since the probe is a pure centered Gaussian state, Wick contraction reduces every fourth-order bosonic moment entering the generator covariances to products of second moments~\cite{Wick1950p268}, and the resulting quantum Fisher matrix depends only on the Gaussian covariance matrix and is diagonal in the canonical passive basis,
		\begin{align}
			\begin{aligned}
				F^{\mathrm{rot}}_{j} =&~ 2 \sinh^{2} \left( 2r_{j} \right), \\
				F^{\mathrm{bsY}}_{jk} =&~ 4 \sinh^{2} \left( r_{j} - r_{k} \right), \\
				F^{\mathrm{bsX}}_{jk} =&~ 4 \sinh^{2} \left( r_{j} + r_{k} \right),
			\end{aligned}
		\end{align}
		under the vacuum covariance convention $\bm{\sigma}_{0} = (1/2) \bm{I}_{2n}$, with $\bm{I}_{2n}$ the $2n$-dimensional identity matrix.
		These weights follow from established Gaussian phase and beam-splitter QFI formulas~\cite{Sparaciari2016p023810,Safranek2016p062313}; we derive their common restricted-root form in Sec.~S6 of the Supplemental Material.
		Each passive channel depends on one of the linear forms $2r_{j}$, $r_{j}-r_{k}$, or $r_{j}+r_{k}$ and becomes locally unidentifiable exactly where that form vanishes.
		
		An invariantly orthonormal basis instead uses $ \widetilde{g}^{\mathrm{rot}}_{j} = \hat g^{\mathrm{rot}}_{j}, \widetilde{g}^{\mathrm{bsX,bsY}}_{jk} = \hat g^{\mathrm{bsX,bsY}}_{jk} / \sqrt{2}$.
		The generator Gram matrix $\bm{M}$, with entries $M_{\mu\nu}=\operatorname{Re}\operatorname{Tr}(\bm{h}_{\mu}\bm{h}_{\nu})$, converts the Fisher matrix to $\widetilde{\bm{F}}=\bm{M}^{-\frac{1}{2}} \bm{F} \bm{M}^{-\frac{1}{2}}$.
		Passive conjugation acts orthogonally in this normalized basis, so the spectrum of $\widetilde{\bm{F}}$ is independent of the passive Bloch--Messiah frame.
		In the angle basis a frame change gives a generally nonorthogonal congruence of $\bm{F}$, which preserves rank but can change ordinary eigenvalues.
		The angle-normalized design applies to canonical probes, or to correspondingly conjugated generators, and does not optimize over arbitrary passive probe orientations against fixed laboratory generators.
		
		%%%%%%%%%%%%%%%%%%%%%%%%%%%%%%%%%%%%%%%%%%%%%%%%%%%%%%%%%%%%%%%%%%%%%%%%%%%%%%
		\section{Restricted-root identifiability}
		\label{sec:Sec3}
		%%%%%%%%%%%%%%%%%%%%%%%%%%%%%%%%%%%%%%%%%%%%%%%%%%%%%%%%%%%%%%%%%%%%%%%%%%%%%%
		
		The linear forms $2r_{j}$, $r_{j}-r_{k}$, and $r_{j}+r_{k}$ define the positive restricted roots of $\mathrm{Sp}(2n,\mathbb{R})/\mathrm{U}(n)$ on the squeezing flat, where $\{\bm{e}_{j}\}$ is the basis of linear forms dual to the diagonal squeezing directions, $\bm{e}_{j}(\bm{r})=r_{j}$,
		\begin{align}
			\Sigma_{+} =&~ \left\{ 2\bm{e}_{j} \right\}_{j=1}^{n} \cup \left\{ \bm{e}_{j}-\bm{e}_{k},\bm{e}_{j}+\bm{e}_{k} \right\}_{1\leq j<k\leq n},
		\end{align}
		with root system $C_{n}$~\cite{Helgason1984,Helgason2001,Knapp1996} and generator correspondence,
		\begin{align}
			\begin{aligned}
				2 \bm{e}_{j} &\quad \longleftrightarrow \quad &\hat{g}^{\mathrm{rot}}_{j}, \\
				\bm{e}_{j}-\bm{e}_{k} &\quad \longleftrightarrow \quad &\hat{g}^{\mathrm{bsY}}_{jk}, \\
				\bm{e}_{j}+\bm{e}_{k} &\quad \longleftrightarrow \quad &\hat{g}^{\mathrm{bsX}}_{jk}.
			\end{aligned}
		\end{align}
		The correspondence is a bijection because the $C_{n}$ system has $n+n(n-1)=n^{2}$ positive roots, matching $\dim\mathrm{u}(n)=n^{2}$, as it must since every restricted root of this symmetric pair has multiplicity one and the centralizer of the squeezing flat in the passive algebra is trivial.
		This generator correspondence allows us to label each passive Fisher weight by its associated positive restricted root and to write all weights in the unified form,
		\begin{align}
			F_{\bm{\alpha}}(\bm{r}) =&~ c_{\bm{\alpha}} \sinh^{2} \left[ \bm{\alpha}(\bm{r}) \right],
			\qquad
			c_{2\bm{e}_{j}} = 2,
			\qquad
			c_{\bm{e}_{j}\pm\bm{e}_{k}} = 4,
		\end{align}
		for $\bm{\alpha}\in\Sigma_{+}$, where $c_{\bm{\alpha}}$ is the generator-normalization coefficient.
		The $\sinh^{2}[\bm{\alpha}(\bm{r})]$ radial dependence and the root-wall stabilizers are standard features of the polar decomposition of this symmetric space~\cite{Helgason1984,Helgason2001}; the restricted-root correspondence identifies their physical sensing channels.
		Since $F_{\bm{\alpha}}(\bm{r})$ vanishes exactly when $\bm{\alpha}(\bm{r})=0$, the passive Fisher matrix loses rank on
		\begin{align}
			\mathcal{A}_{C_{n}} =&~ \bigcup_{\bm{\alpha} \in \Sigma_{+}} \ker \bm{\alpha}.
		\end{align}
		Each hyperplane $\ker\bm{\alpha}$ identifies the passive generator that loses local identifiability there, while every squeezing spectrum outside $\mathcal{A}_{C_{n}}$ gives full passive Fisher rank.
		
		The positive restricted roots also organize the action of squeezing on the passive generators in phase space.
		We write $\bm{K}_{\bm{\alpha}}$ for the phase-space matrix of the passive generator associated with $\bm{\alpha}\in\Sigma_{+}$ and $\bm{P}_{\bm{\alpha}}$ for its corresponding active direction.
		The phase-space matrix of the squeezing exponent,
		\begin{align}
			\bm{A}(\bm{r}) =&~
			\begin{pmatrix}
				-\bm{R} & \bm{0} \\
				\bm{0} & \bm{R}
			\end{pmatrix},
			\qquad
			\bm{R} = \mathrm{diag} \left( r_{1},\ldots,r_{n} \right),
		\end{align}
		acts within each pair $\left(\bm{K}_{\bm{\alpha}},\bm{P}_{\bm{\alpha}}\right)$ as
		\begin{align}
			\begin{aligned}
				\left[ \bm{A}(\bm{r}),\bm{K}_{\bm{\alpha}} \right] =&~ \bm{\alpha}(\bm{r}) \bm{P}_{\bm{\alpha}}, \\
				\left[ \bm{A}(\bm{r}),\bm{P}_{\bm{\alpha}} \right] =&~ \bm{\alpha}(\bm{r}) \bm{K}_{\bm{\alpha}},
			\end{aligned}
		\end{align}
		giving the conjugated passive direction,
		\begin{align}
			e^{-\bm{A}(\bm{r})} \bm{K}_{\bm{\alpha}} e^{\bm{A}(\bm{r})} =&~
			\cosh \left[ \bm{\alpha}(\bm{r}) \right] \bm{K}_{\bm{\alpha}}
			- \sinh \left[ \bm{\alpha}(\bm{r}) \right] \bm{P}_{\bm{\alpha}},
		\end{align}
		with active component proportional to $\sinh \left[ \bm{\alpha}(\bm{r}) \right]$.
		
		\begin{figure*}[t]
			\centering
			\includegraphics[width=150mm]{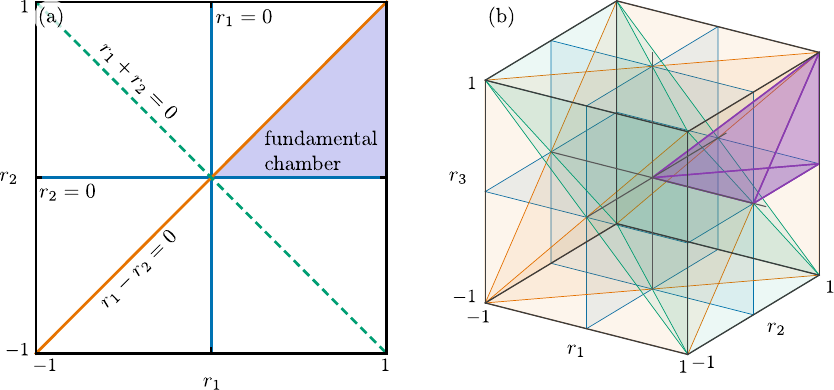}
			\caption{ Restricted-root arrangements in squeezing-spectrum space.
				(a) Arrangement for the $C_{2}$ restricted-root system.
				The four singular lines $r_{1}=0$, $r_{2}=0$, $r_{1}=r_{2}$, and $r_{1}+r_{2}=0$ correspond to the vanishing Fisher weights of the $\mathrm{rot}_{1}$, $\mathrm{rot}_{2}$, $\mathrm{bsY}_{12}$, and $\mathrm{bsX}_{12}$ channels.
				The shaded region indicates the fundamental Weyl chamber.
				(b) Arrangement for the $C_{3}$ restricted-root system.
				The singular planes partition the three-dimensional squeezing space into Weyl chambers related by permutations and sign changes of the squeezing spectrum.
				Each plane identifies the passive generator that becomes locally unidentifiable. }
			\label{fig:Fig1}
		\end{figure*}
		
		Figure~\ref{fig:Fig1} shows the restricted-root arrangement for the first two nontrivial cases, illustrating how each singular hyperplane partitions squeezing-spectrum space into regular regions separated by Fisher-null walls.
		The full squeezing flat, $\bm{r}\in\mathbb{R}^{n}$, is divided by the singular hyperplanes $r_{j}=0$, $r_{j}-r_{k}=0$, and $r_{j}+r_{k}=0$.
		Passive unitary changes of reference frame generated by mode permutations and phase rotations relate squeezing spectra that differ only by permutations and sign changes of their entries, corresponding to the restricted Weyl-group action~\cite{Helgason2001,Bourbaki2002,Knapp1996}, so we restrict the probe design to the fundamental chamber, $r_{1} \geq \cdots \geq r_{n} \geq 0$.
		Within this chamber, the independent singular walls are the vanishing-squeezing boundary $r_{n}=0$ and the adjacent equal-squeezing walls $r_{j}=r_{j+1}$, $j=1,\ldots,n-1$.
		The sum-root condition $r_{j}+r_{k}=0$ occurs only when both corresponding squeezing parameters vanish and is already contained in the vanishing-squeezing boundary.
		A vanishing squeezing parameter removes local phase sensitivity from the corresponding unsqueezed mode, while an equal-squeezing wall $r_{j}=r_{k}$ removes sensitivity to the corresponding $\mathrm{bsY}$ beam-splitter quadrature.
		Equal squeezing is also used in the sequential phase model~\cite{Frigerio2026p2540001}, but its phase redundancy is independent of the squeezing spectrum and is distinct from this covariance wall.
		The classification is local at the identity, $\bm{\theta}=\bm{0}$.
		At another passive operating point it describes the full tangent model in a correspondingly conjugated generator basis.
		A fixed nonlinear parameterization can introduce an additional Jacobian, so its coordinate Fisher matrix need not retain these diagonal entries or eigenvalues.
		
		%%%%%%%%%%%%%%%%%%%%%%%%%%%%%%%%%%%%%%%%%%%%%%%%%%%%%%%%%%%%%%%%%%%%%%%%%%%%%%
		\section{Probe design}
		\label{sec:Sec4}
		%%%%%%%%%%%%%%%%%%%%%%%%%%%%%%%%%%%%%%%%%%%%%%%%%%%%%%%%%%%%%%%%%%%%%%%%%%%%%%
		
		Our restricted-root classification identifies the squeezing spectra for which one or more passive channels become locally unidentifiable.
		Probe design requires choosing a spectrum $\bm{r}^{\star}$ inside the fundamental chamber that remains separated from every singular wall.
		We first maximize the smallest positive-root evaluation at fixed Euclidean norm,
		\begin{align}
			\bm{r}^{\star} =&~
			\arg\max_{\substack{\lVert\bm{r}\rVert=1\\
					r_{1}\geq\cdots\geq r_{n}\geq0}}
			\min_{\bm{\alpha}\in\Sigma_{+}}
			\bm{\alpha}(\bm{r}).
		\end{align}
		The algebraic root margin
		\begin{align}
			m(\bm{r}) =&~ \min_{\bm{\alpha}\in\Sigma_{+}} \bm{\alpha}(\bm{r})
		\end{align}
		measures the smallest argument entering the passive Fisher weights and vanishes on the singular arrangement.
		It is an algebraic root margin because the $C_{n}$ roots have two lengths,
		\begin{align}
			\left\lvert 2\bm{e}_{j}\right\rvert^{2} =&~ 4,
			\qquad
			\left\lvert\bm{e}_{j}\pm\bm{e}_{k}\right\rvert^{2} = 2.
		\end{align}
		Its unique maximizer lies along the dual Weyl vector $\bm{\rho}^{\vee}$, the half-sum of the positive coroots~\cite{Bourbaki2002,Knapp1996},
		\begin{align}
			\bm{r}^{\star} \propto \bm{\rho}^{\vee} =&~ \frac{1}{2} \left( 2 n - 1, 2 n - 3, \ldots, 3, 1 \right),
		\end{align}
		which gives the squeezing ratios $3:1$, $5:3:1$, and $7:5:3:1$ for $n=2,3,4$.
		
		Our scale-free optimization identifies the chamber direction that maximizes the smallest root argument, but it does not yet impose the physical squeezing resource or maximize the weakest Fisher weight.
		The physical probe-design problem holds the mean photon number constant,
		\begin{align}
			N =&~ \sum_{j=1}^{n}\sinh^{2} r_{j},
		\end{align}
		and seeks the spectrum that maximizes
		\begin{align}
			w(\bm{r}) =&~ \min_{\bm{\alpha}\in\Sigma_{+}}
			c_{\bm{\alpha}}\sinh^{2}\left[\bm{\alpha}(\bm{r})\right].
		\end{align}
		
		Because the passive Fisher matrix is diagonal in the canonical passive basis, its spectrum is the set of Fisher weights, and the objective $w(\bm{r})$ is its smallest eigenvalue,
		\begin{align}
			w(\bm{r}) =&~ \lambda_{\min}\left[ \bm{F}_{\mathrm{passive}}(\bm{r}) \right].
		\end{align}
		Maximizing $w$ at fixed photon number applies the E-optimal criterion of experimental design~\cite{Kiefer1974p849,Pukelsheim1993} to this family of quantum Fisher matrices in the canonical angle normalization.
		For a general Fisher matrix the weakest-channel objective only bounds the smallest eigenvalue from above, $\min_{\mu}F_{\mu\mu}\geq\lambda_{\min}(\bm{F})$ by the Rayleigh principle; the two coincide for the canonical squeezed probe because the restricted-root decomposition diagonalizes the passive response.
		For a positive-definite Fisher matrix,
		\begin{align}
			\max_{\lVert\bm{u}\rVert=1}\bm{u}^{\mathrm{T}}\bm{F}^{-1}\bm{u}
			=&~ \frac{1}{\lambda_{\min}(\bm{F})}.
		\end{align}
		The design minimizes the largest directional inverse-QFI bound per copy.
		A basis generator attains this maximum; when the smallest eigenvalue is degenerate, any unit vector in its eigenspace also attains it.
		Simultaneous attainability of the multiparameter bound requires a separate compatibility analysis.
		
		For $n\geq2$ in the fundamental chamber, the smallest squeezing $r_{n}$ determines the weakest local-rotation channel, while the smallest adjacent gap,
		\begin{align}
			\delta =&~ \min_{1\leq j\leq n-1}\left(r_{j}-r_{j+1}\right),
		\end{align}
		determines the weakest difference-root channel.
		Every sum-root channel satisfies
		\begin{align}
			4\sinh^{2}\left(r_{i}+r_{j}\right) \geq&~
			2\sinh^{2}\left(2r_{n}\right),
		\end{align}
		so a sum-root channel cannot lower the minimum set by the phase and difference-root channels, and the physical objective reduces to
		\begin{align}
			w(\bm{r}) =&~ \min \left\{ 2 \sinh^{2} \left(2 r_{n} \right), 4 \sinh^{2} \delta \right\}.
		\end{align}
		
		For a prescribed weakest weight $w$, we obtain the smallest photon budget by saturating the smallest squeezing and every adjacent gap.
		With the smallest squeezing $a=r_{n}>0$ and common adjacent gap $b>0$, the optimal spectrum has the arithmetic progression,
		\begin{align}
			r_{j} =&~ a+(n-j)b, \qquad
			2 \sinh^{2} \left( 2 a \right) = 4 \sinh^{2} b = w.
		\end{align}
		The photon budget increases strictly with every $r_{j} \geq 0$, so the same spectrum uniquely maximizes $w$ at fixed $N$.
		The exact optimum depends on the available photon number through the relation
		\begin{align}
			\sinh b =&~ \frac{1}{\sqrt{2}} \sinh \left( 2 a \right).
		\end{align}
		As the photon budget increases, $a \rightarrow \infty$ and $b/a \rightarrow 2$, and the corresponding arithmetic progression approaches the dual-Weyl ray, $\bm{r} \propto ( 2 n - 1, 2 n - 3, \ldots, 3, 1 )$, recovering the solution of the scale-free optimization problem.
		For $n=1$, the single phase has $w=2\sinh^{2}(2r_{1})$ and $r_{1}=\operatorname{arcsinh}\sqrt{N}$.
		Figure~\ref{fig:Fig2} illustrates both optimization problems.
		The left panel shows the finite-resource objective over the $C_{2}$ fundamental chamber, while the right panel shows how the exact finite-photon optimum approaches the dual-Weyl direction as the available squeezing resource increases.
		
		For comparison, the invariantly orthonormal basis gives $\widetilde F_{\bm{\alpha}}=2\sinh^{2}[\bm{\alpha}(\bm{r})]$ for every positive root.
		Its worst weight is $\widetilde w=2\sinh^{2}m$, where $m=\min_{\bm{\alpha}\in\Sigma_+}\bm{\alpha}(\bm{r})$ in the fundamental chamber.
		At fixed $N>0$ the unique invariantly normalized optimum is $r_j=(2n-2j+1)a$, with $\sum_{j=1}^{n}\sinh^{2}[(2n-2j+1)a]=N$.
		We prove this exact dual-Weyl result in Sec.~S7 of the Supplemental Material.
		Both optimal spectra, in the angle and invariant normalizations, have equal adjacent gaps.
		In the angle normalization, the ratio of the adjacent gap to the smallest squeezing approaches two as the photon number increases.
		
		\begin{figure*}[t]
			\centering
			\includegraphics[width=150mm]{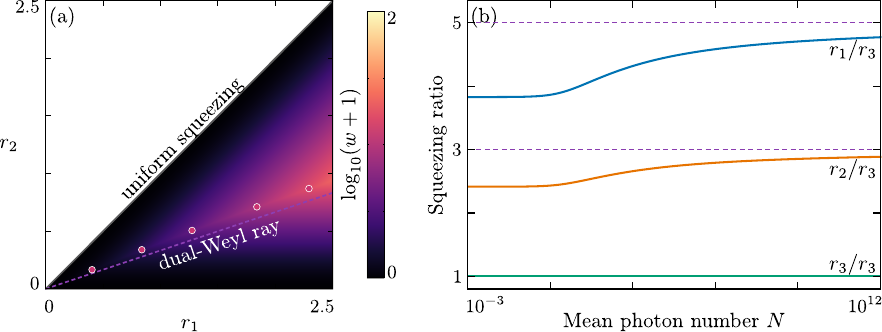}
			\caption{ E-optimal (maximin) probe design in the canonical angle normalization.
				(a) The quantity $\log_{10}[1+w(\bm{r})]$ over the $C_{2}$ fundamental chamber. Markers show optima at $N=0.2,1,3,10,25$. The uniform-squeezing ray lies on the Fisher-null wall, whereas the finite-photon optimum remains in the chamber interior and moves toward the dual-Weyl direction.
				(b) Normalized squeezing ratios for the exact optimal spectrum of the three-mode problem. The range $10^{-3}\le N\le10^{12}$ displays the slow convergence to the dual-Weyl ratios $5:3:1$ (dashed lines). These dashed ratios also give the exact optimum at every $N$ under the invariant generator normalization. }
			\label{fig:Fig2}
		\end{figure*}

		%%%%%%%%%%%%%%%%%%%%%%%%%%%%%%%%%%%%%%%%%%%%%%%%%%%%%%%%%%%%%%%%%%%%%%%%%%%%%%
		\section{Local sensing versus network calibration}
		\label{sec:Sec5}
		%%%%%%%%%%%%%%%%%%%%%%%%%%%%%%%%%%%%%%%%%%%%%%%%%%%%%%%%%%%%%%%%%%%%%%%%%%%%%%
		
		For $n$ independent local phases, the Fisher matrix is diagonal,
		\begin{align}
			\bm{F}_{\mathrm{phase}} =&~ \mathrm{diag} \left( 2 \sinh^{2} \left( 2 r_{1} \right), \ldots, 2 \sinh^{2} \left( 2r_{n} \right)
			\right),
		\end{align}
		and the corresponding A-optimal cost is the sum of the individual inverse-QFI sensitivities~\cite{Kiefer1974p849,Pukelsheim1993},
		\begin{align}
			C_{\mathrm{phase}}(\bm{r}) =&~ \mathrm{Tr}\left(\bm{F}_{\mathrm{phase}}^{-1}\right)
			= \sum_{j=1}^{n} \frac{1}{2\sinh^{2}\left(2r_{j}\right)}.
		\end{align}
		At fixed mean photon number, this cost is uniquely minimized by the uniform squeezing spectrum,
		\begin{align}
			r_{1} = \cdots =r_{n} =&~ \operatorname{arcsinh} \sqrt{\frac{N}{n}}.
		\end{align}
		Local-phase sensing favors equal distribution of the squeezing resource, in contrast with the nonuniform arithmetic progression selected by the E-optimal network objective.
		This optimum follows from the single-mode squeezed-vacuum phase QFI~\cite{Monras2006p033821} and convex resource allocation.
		The phase and network objectives apply to different parameter spaces, so their comparison concerns both the estimation task and the scalar cost.
		For the full passive model with $n\geq2$, uniform squeezing gives $\lambda_{\min}=0$ and $\det\bm{F}=0$, while $\mathrm{Tr}(\bm{F}^{-1})$ diverges on approach to the wall.
		A finite trace $\mathrm{Tr}\,\bm{F}$ alone does not ensure that every parameter is identifiable.
		
		For passive-network calibration, each mode pair has two beam-splitter quadratures with individual inverse-QFI costs,
		\begin{align}
			\begin{aligned}
				\frac{1}{F^{\mathrm{bsY}}_{ij}} =&~ \frac{1}{4 \sinh^{2} \left( r_{i} -r_{j} \right)}, \\
				\frac{1}{F^{\mathrm{bsX}}_{ij}} =&~ \frac{1 }{4 \sinh^{2} \left( r_{i} + r_{j} \right)}.
			\end{aligned}
		\end{align}
		Uniform squeezing places the probe on every equal-squeezing wall,
		\begin{align}
			F^{\mathrm{bsY}}_{ij} =&~ 0, \qquad 1 \leq i < j \leq n,
		\end{align}
		so all $n(n-1)/2$ $\mathrm{bsY}$ coupling directions become locally unidentifiable.
		The common intersection of these walls is the uniform-squeezing ray and has codimension $n-1$ in squeezing-spectrum space.
		No measurement can recover a parameter direction whose QFI vanishes~\cite{Helstrom1976,Braunstein1994p3439}.
		A probe required to sense local phases while retaining sensitivity to all passive-network couplings must depart from the uniform squeezing spectrum.
		
		Changing the squeezing spectrum removes the covariance-wall degeneracy.
		In contrast, an intermediate squeezer separates successive phase encodings by changing their effective generators~\cite{Frigerio2026p2540001}; unequal input squeezing alone cannot remove their bare phase-sum redundancy.
		The inverse-QFI network weights characterize the individual coupling directions; whether their bounds can be attained simultaneously is a separate multiparameter-compatibility question.
		
		For $n=3$ and $N=2$, the E-optimal network spectrum $\bm{r}_{\mathrm{network}} \simeq \left( 1.021, 0.642, 0.262 \right)$ gives a local phase cost $C_{\mathrm{phase}}\left(\bm{r}_{\mathrm{network}}\right) \simeq 1.87$.
		Using the ideal quadrature-variance convention $s_{j}=20r_{j}/\ln 10$, the optimal spectrum corresponds to
		$\bm{s}_{\mathrm{network}}\simeq(8.87,5.57,2.28)\,\mathrm{dB}$.
		These values describe the input probe before loss and other experimental imperfections and are below the $10\,\mathrm{dB}$ squeezing demonstrated with cavity and waveguide sources~\cite{Hagemann2024p7954,Hirota2026p7958}.
		A cavity source has reached $12.6\,\mathrm{dB}$ and maintained more than $10\,\mathrm{dB}$ over $50\,\mathrm{h}$ with a $96.6\%$ duty cycle~\cite{Shajilal2022p37213}.
		
		At this spectrum, $\lambda_{\min}\simeq0.6033$ is attained by $\mathrm{rot}_{3}$, $\mathrm{bsY}_{12}$, and $\mathrm{bsY}_{23}$, consistent with the balance condition between the weakest phase and difference-root weights.
		The largest weight is $\lambda_{\max}\simeq28.68$ on $\mathrm{rot}_{1}$, giving the spectral condition number $\kappa\simeq47.53$.
		The finite value establishes invertibility but does not by itself imply good conditioning.
		E-optimality maximizes $\lambda_{\min}$, whereas condition-number minimization concerns the ratio $\kappa=\lambda_{\max}/\lambda_{\min}$.
		Indeed, along the E-optimal family, $\ln\kappa=4(n-1)b+\mathcal{O}(1)$ in the strong-resource limit.
		A fully identifiable probe can still develop the broad Fisher hierarchy associated with classical sloppiness~\cite{Gutenkunst2007pe189,Transtrum2015p010901}.
		Minimizing $\kappa$ at fixed photon number is a separate design problem.
		
		At the same photon number, the uniform spectrum $\bm{r}_{\mathrm{phase}} \simeq \left( 0.745, 0.745, 0.745 \right)$ minimizes the local-phase A-optimal cost, $C_{\mathrm{phase}}\left(\bm{r}_{\mathrm{phase}}\right) = 27/80 \simeq 0.34$.
		This spectrum lies on all three equal-squeezing walls, so every $\mathrm{bsY}$ coupling direction is locally unidentifiable.
		
		The nonuniform spectrum restores sensitivity to every passive coupling direction, but it increases the local-phase A-optimal cost, meaning a worse aggregate precision for estimating the local phases.
		Figure~\ref{fig:Fig3} compares the passive Fisher spectra associated with the two optimal probes.
		The E-optimal spectrum gives every passive channel a positive weight, whereas uniform squeezing retains sensitivity to local rotations and sum-root beam-splitter channels but eliminates every difference-root channel.
		
		\begin{figure*}[t]
			\centering
			\includegraphics[width=150mm]{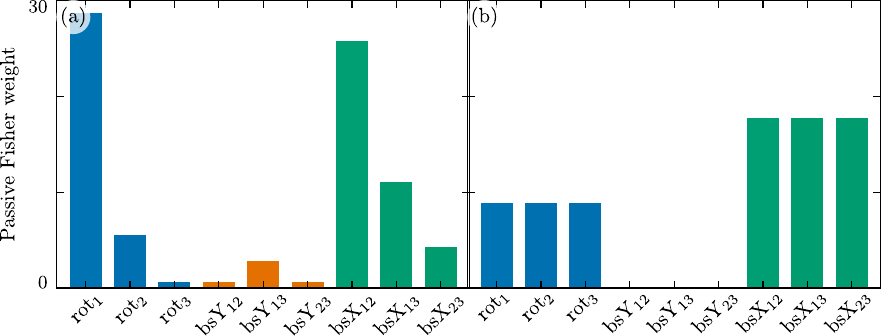}
			\caption{ Passive Fisher spectra for the two competing probe-design objectives at $n=3$ and mean photon number $N=2$.
				(a) The E-optimal network probe redistributes the squeezing resource according to the optimal arithmetic progression.
				Every passive channel acquires a nonzero Fisher weight, although the aggregate precision for local-phase estimation decreases.
				(b) Uniform squeezing minimizes the local-phase A-optimal cost but places the probe on every equal-squeezing wall, causing all three difference-root beam-splitter channels to become Fisher-null. }
			\label{fig:Fig3}
		\end{figure*}
		
		%%%%%%%%%%%%%%%%%%%%%%%%%%%%%%%%%%%%%%%%%%%%%%%%%%%%%%%%%%%%%%%%%%%%%%%%%%%%%%
		\section{Measurement compatibility}
		\label{sec:Sec6}
		%%%%%%%%%%%%%%%%%%%%%%%%%%%%%%%%%%%%%%%%%%%%%%%%%%%%%%%%%%%%%%%%%%%%%%%%%%%%%%
		
		Passive identifiability and measurement compatibility describe distinct properties of the local metrological model.
		Our restricted roots determine whether a passive generator produces a nonzero tangent on the canonical squeezed probe,
		while compatibility depends on the expectation values of generator commutators.
		
		For a pure unitary model, a pair of parameters satisfies the standard weak-commutativity condition when
		\begin{align}
			\mathcal{C}_{\mu\nu} =&~ \frac{1}{2i} \left\langle \left[ \hat{g}_{\mu}, \hat{g}_{\nu} \right] \right\rangle = 0.
		\end{align}
		For an identifiable pure-state submodel, vanishing mean commutators for every parameter pair are necessary and sufficient for a measurement to attain the SLD quantum Cram\'er--Rao matrix bound at the operating point, with locally unbiased estimators and unrestricted measurements~\cite{Matsumoto2002p3111,Ragy2016p052108}.
		We use measurement compatibility in this local sense, which depends on the probe and operating point~\cite{Nichols2018p012114,Imai2026p2602.12097}.
		A Fisher-null generator weakly commutes trivially with every direction, so we assess compatibility and identifiability jointly.
		
		For the canonical squeezed probe, all passive-generator expectation values vanish except those of the local number operators,
		\begin{align}
			\left\langle \hat{g}_{j}^{\mathrm{rot}} \right\rangle =&~ \sinh^{2} r_{j}.
		\end{align}
		The only passive commutator with a nonzero expectation value couples the two beam-splitter quadratures of the same mode pair,
		\begin{align}
			\left[ \hat{g}_{jk}^{\mathrm{bsX}}, \hat{g}_{jk}^{\mathrm{bsY}} \right] =&~ 2 i \left( \hat{g}_{j}^{\mathrm{rot}} - \hat{g}_{k}^{\mathrm{rot}} \right),
		\end{align}
		giving a mean commutator
		\begin{align}
			\begin{aligned}
				\mathcal{C}_{\mathrm{bsX}_{jk},\mathrm{bsY}_{jk}} =&~ \sinh \left( r_{j} + r_{k} \right) \sinh \left( r_{j} - r_{k} \right),
			\end{aligned}
		\end{align}
		which vanishes precisely when $r_{j}=r_{k}$ or $r_{j}=-r_{k}$.
		The passive mean-commutator tensor decomposes into one antisymmetric block for each mode pair,
		\begin{align}
			\bm{\mathcal{C}}_{jk} =&~
			\begin{pmatrix}
				0 & \mathcal{C}_{\mathrm{bsX}_{jk},\mathrm{bsY}_{jk}} \\
				-\mathcal{C}_{\mathrm{bsX}_{jk},\mathrm{bsY}_{jk}} & 0
			\end{pmatrix},
		\end{align}
		with null rows and columns for the local rotations.
		
		The restricted-root Fisher weights give the exact relation,
		\begin{align}
			\left( \mathcal{C}_{\mathrm{bsX}_{jk},\mathrm{bsY}_{jk}} \right)^{2} =&~ \frac{1}{16} F_{jk}^{\mathrm{bsX}} F_{jk}^{\mathrm{bsY}}.
		\end{align}
		At a regular squeezing spectrum, $\bm{\alpha}(\bm{r}) \neq 0$ for every $\bm{\alpha} \in \Sigma_{+}$, both Fisher weights are nonzero and
		\begin{align}
			\mathcal{C}_{\mathrm{bsX}_{jk},\mathrm{bsY}_{jk}} \neq&~ 0,
		\end{align}
		so the two identifiable beam-splitter quadratures of every mode pair are incompatible.
		The canonical centered squeezed probe has no squeezing spectrum for which both beam-splitter quadratures of one mode pair are simultaneously identifiable and weakly compatible.
		
		In fact, the incompatibility is maximal in a normalized quantum-geometric sense.
		The positive semidefinite tangent Gram matrix is $\bm{Q}=\bm{F}/4+i\bm{\mathcal{C}}$, so each two-parameter principal minor gives
		\begin{align}
			\mathcal C_{\mu\nu}^{2}
			\leq&~ \frac{F_{\mu\mu}F_{\nu\nu}-F_{\mu\nu}^{2}}{16}.
		\end{align}
		Each identifiable beam-splitter pair saturates this inequality, since its block of $\bm{Q}$ has rank one and its centered state tangents are complex collinear.
		For $\bm{F}>0$, we define the dimensionless incompatibility measure~\cite{Carollo2019p094010}
		\begin{align}
			R =&~ \left\lVert i\bm{F}^{-1/2}\bm{D}\bm{F}^{-1/2}\right\rVert_{\mathrm{op}},
			\qquad \bm{D}=4\bm{\mathcal{C}},
		\end{align}
		where $\bm{D}$ is the SLD mean-commutator matrix and $\lVert\cdot\rVert_{\mathrm{op}}$ is the operator norm.
		Each identifiable beam-splitter pair attains the maximum $R=1$ because it saturates the quantum-geometric bound; the full regular model likewise has $R=1$ for $n\geq2$.
		A positive cost matrix weights the estimator covariance to define a scalar estimation error.
		The SLD cost is the corresponding weighted trace of $\bm{F}^{-1}$, while the Holevo cost gives the tighter attainable bound that accounts for measurement incompatibility~\cite{DemkowiczDobrzanski2020p363001,Chang2026p126}.
		For pure-state models the Holevo bound is locally attainable with measurements on individual copies, and its cost lies between the SLD cost and twice that cost for the same weight matrix~\cite{Matsumoto2002p3111,Albarelli2020p1911.11036}.
		The value $R=1$ does not imply that every cost matrix attains the upper limit.
		
		At a regular squeezing spectrum, we can construct a maximal pairwise weakly compatible passive submodel with all $n$ local rotations and one beam-splitter quadrature from each mode pair, for a total of
		\begin{align}
			n + \frac{n(n-1)}{2} =&~ \frac{n(n+1)}{2}
		\end{align}
		identifiable parameters.
		The independent choice of $\hat{g}_{jk}^{\mathrm{bsX}}$ or $\hat{g}_{jk}^{\mathrm{bsY}}$ for every mode pair gives
		\begin{align}
			2^{n(n-1)/2}
		\end{align}
		basis-adapted submodels of this maximal dimension, while the full $n^{2}$-parameter passive model is identifiable but not weakly compatible away from the restricted-root arrangement.
		This dimension is maximal among all real tangent subspaces, not only selections of basis generators.
		Indeed, at a regular spectrum $\operatorname{rank}\bm{\mathcal{C}}=n(n-1)$, and an antisymmetric form on an $n^{2}$-dimensional space has maximal isotropic dimension $n^{2}-\operatorname{rank}\bm{\mathcal{C}}/2=n(n+1)/2$.
		These basis-adapted submodels attain this bound. For $n=1$, the single identifiable phase is compatible and $R=0$.
		
		%%%%%%%%%%%%%%%%%%%%%%%%%%%%%%%%%%%%%%%%%%%%%%%%%%%%%%%%%%%%%%%%%%%%%%%%%%%%%%
		\section{Conclusion}
		\label{sec:Sec7}
		%%%%%%%%%%%%%%%%%%%%%%%%%%%%%%%%%%%%%%%%%%%%%%%%%%%%%%%%%%%%%%%%%%%%%%%%%%%%%%
		
		We showed that sloppiness in passive Gaussian metrology is governed by the $C_{n}$ restricted-root structure of the squeezing spectrum.
		Each singular configuration identifies a specific phase-rotation or beam-splitter direction that becomes locally unidentifiable, while the corresponding root evaluation controls the approach to Fisher rank loss and the rate at which the associated Fisher weight vanishes.
		
		At fixed positive photon number, in canonical angle coordinates or their conjugated basis, the E-optimal spectrum is unique within the fundamental chamber, forms an arithmetic progression of squeezing magnitudes, and approaches the consecutive-odd-integer dual-Weyl direction in the strong-resource limit.
		With an invariant generator norm the optimum is exactly dual-Weyl at every photon number and its Fisher spectrum is frame independent.
		This nonuniform design contrasts with the A-optimal local-phase probe, for which uniform squeezing is optimal but makes one beam-splitter quadrature unidentifiable for every mode pair.
		
		Identifiability and measurement compatibility remain distinct but are linked within each mode pair.
		Whenever both beam-splitter quadratures are identifiable, they saturate the quantum-geometric incompatibility bound; weak commutativity occurs only when one is Fisher-null.
		At a regular spectrum, $n(n+1)/2$ is the largest possible dimension of a compatible passive tangent submodel.
		Our restricted-root viewpoint unifies passive rank loss, Fisher weight assignment, optimal probe design in the Kiefer sense, and compatibility in a single algebraic framework.
		Two design problems remain open within the same framework, namely minimizing the condition number of the passive Fisher matrix at fixed photon number, and the joint optimization over squeezing spectrum and coherent displacement, in which both resources compete for the same photon budget.
		
		%%%%%%%%%%%%%%%%%%%%%%% Backmatter %%%%%%%%%%%%%%%%%%%%%%%%%
		
		\section*{Funding}
		The authors received no specific funding for this work.
		
		\section*{Acknowledgments}
		B.~M.~R.~L. thanks Jacinta Alderete Galan for providing daycare support throughout this work.
		He also acknowledges support and hospitality as an affiliate visiting colleague at the Department of Physics and Astronomy, University of New Mexico, as well as fruitful discussions with F.~E. Becerra.
		
		\section*{Disclosures}
		The authors declare no conflicts of interest.
		
		\section*{Tool and Computational Resource}
		We used OpenAI Codex and Anthropic Claude Fable 5.1 (High effort) to produce adversarial review of the manuscript, identify potential gaps in the arguments, and improve clarity and grammar.
		We reviewed the outputs, verified the mathematical arguments and calculations, and made all final decisions on the manuscript.
		We take full responsibility for the accuracy, integrity, and attribution of the final work.
		
		\section*{Data Availability Statement}
		The data that support the findings of this study are available from the corresponding author upon reasonable request.
		
		%%%%%%%%%%%%%%%%%%%%%%% References %%%%%%%%%%%%%%%%%%%%%%%%%

	%	\putbib[references]
	

	\end{bibunit}
	
	\clearpage
	
	%%%%%%%%%%%%%%%%%%%%%%% Supplemental Material %%%%%%%%%%%%%%%%%%%%%%%%%
	
	\setcounter{section}{0}
	\setcounter{subsection}{0}
	\setcounter{equation}{0}
	\setcounter{figure}{0}
	\setcounter{table}{0}
	
	\renewcommand{\thesection}{S\arabic{section}}
	\renewcommand{\thesubsection}{S\arabic{section}.\arabic{subsection}}
	\renewcommand{\theequation}{S\arabic{equation}}
	\renewcommand{\thefigure}{S\arabic{figure}}
	\renewcommand{\thetable}{S\arabic{table}}
	
	% Distinct PDF destinations for counters reset in the Supplemental Material.
	\renewcommand{\theHsection}{supp.\arabic{section}}
	\renewcommand{\theHsubsection}{supp.\arabic{section}.\arabic{subsection}}
	\renewcommand{\theHequation}{supp.\arabic{equation}}
	\renewcommand{\theHfigure}{supp.\arabic{figure}}
	\renewcommand{\theHtable}{supp.\arabic{table}}
	
	% All theorem-like environments share the `theorem' counter, so only that one
	% exists and only that one may be reset or renumbered here.
	\setcounter{theorem}{0}
	\renewcommand{\thetheorem}{S\arabic{theorem}}
	\renewcommand{\theHtheorem}{supp.\arabic{theorem}}

	\begin{bibunit}[apsrev4-2]
		
		\begin{center}
			
			{\large\bfseries Supplemental Material for\\[0.5em]
				``Root-system structure of sloppiness in passive Gaussian metrology''}
			
			\vspace{1.5em}
			
			A. Kafuri$^{1}$ and B.~M. Rodr\'iguez-Lara$^{2}$
			
			\vspace{0.75em}
			
			{\small
				$^{1}$Departamento de F\'isica, Centro de Investigaci\'on y de Estudios Avanzados del IPN,
				P.O. Box 14-740, CP. 07000, Ciudad de M\'exico, Mexico.\\[0.5em]
				$^{2}$Universidad Polit\'ecnica Metropolitana de Hidalgo,
				Tolcayuca, Hidalgo 43860, Mexico.}
			
		\end{center}
		
		\vspace{1em}
		
		We develop the quadratic bosonic framework and the ordinary and restricted root constructions in Secs.~\ref{sec:S1}--\ref{sec:S6}, deriving the passive Fisher weights used in the probe-optimization proofs of Sec.~\ref{sec:S7} and the compatibility classification of Sec.~\ref{sec:S8}.
		These results allow us to distinguish sequential phase redundancy from covariance-wall degeneracy in the two-mode reduction of Sec.~\ref{sec:S8b} and to examine phase sensing, network calibration, displacement, and isotropic additive noise in Sec.~\ref{sec:S9}.
		The comparison in Sec.~\ref{sec:S10} summarizes the defining structures and physical roles of the ordinary and restricted roots.
		
		%%%%%%%%%%%%%%%%%%%%%%%%%%%%%%%%%%%%%%%%%%%%%%%%%%%%%%%%%%%%%%%%%%%%%%%%%%%%%%
		\section{Quadratic bosonic Gaussian framework}
		\label{sec:S1}
		%%%%%%%%%%%%%%%%%%%%%%%%%%%%%%%%%%%%%%%%%%%%%%%%%%%%%%%%%%%%%%%%%%%%%%%%%%%%%%
		
		A pure zero-mean Gaussian probe admits a Bloch--Messiah factorization into passive transformations and a canonical diagonal squeezing transformation that the spectrum $\bm{r}$ characterizes.
		Choosing the canonical squeezing basis leaves a passive $\mathrm{U}(n)$ orbit that phase rotations and beam-splitter transformations generate.
		The Gaussian contraction matrix determines the Fisher and mean-commutator tensors on this orbit.
		
		Two distinct root constructions of abstract type $C_{n}$ organize the resulting geometry.
		The ordinary roots of the complexified algebra $\mathrm{sp}(2n,\mathbb{C})$, defined relative to the number Cartan, organize quadratic-generator commutators and measurement compatibility.
		The restricted roots of the symmetric pair $\mathrm{Sp}(2n,\mathbb{R})/\mathrm{U}(n)$, defined relative to the diagonal squeezing flat, control the passive Fisher weights, their singular hyperplanes, and the associated probe-optimization problem.
		
		%%%%%%%%%%%%%%%%%%%%%%%%%%%%%%%%%%%%%%%%%%%%%%%%%%%%%%%%%%%%%%%%%%%%%%%%%%%%%%
		\subsection{Bosonic phase space}
		%%%%%%%%%%%%%%%%%%%%%%%%%%%%%%%%%%%%%%%%%%%%%%%%%%%%%%%%%%%%%%%%%%%%%%%%%%%%%%
		
		We consider $n$ bosonic modes with canonical quadratures
		\begin{align}
			\hat{\bm{\chi}} =&~ \left( \hat{q}_{1},\ldots,\hat{q}_{n},\hat{p}_{1},\ldots,\hat{p}_{n} \right)^{\mathrm{T}},
		\end{align}
		satisfying the canonical commutation relations with symplectic form $\bm{\Omega}$,
		\begin{align}
			\left[ \hat{\chi}_{a},\hat{\chi}_{b} \right] =&~ i\Omega_{ab},
			\qquad
			\bm{\Omega} =
			\begin{pmatrix}
				\bm{0} & \bm{I}_{n} \\
				-\bm{I}_{n} & \bm{0}
			\end{pmatrix},
		\end{align}
		where $\bm{I}_{n}$ is the $n$-dimensional identity matrix.
		The annihilation and creation operators are
		\begin{align}
			\hat{a}_{j} =&~ \frac{1}{\sqrt{2}} \left( \hat{q}_{j}+i\hat{p}_{j} \right),
			\qquad
			\hat{a}_{j}^{\dagger} = \frac{1}{\sqrt{2}} \left( \hat{q}_{j}-i\hat{p}_{j} \right).
		\end{align}
		
		A Gaussian state has first moments and a covariance matrix,
		\begin{align}
			d_{a} =&~ \left\langle \hat{\chi}_{a} \right\rangle,
			\qquad
			\sigma_{ab} = \frac{1}{2} \left\langle \left\{ \hat{\chi}_{a}-d_{a},\hat{\chi}_{b}-d_{b} \right\} \right\rangle,
		\end{align}
		and the vacuum has covariance
		\begin{align}
			\bm{\sigma}_{0} =&~ \frac{1}{2}\bm{I}_{2n}.
		\end{align}
		
		%%%%%%%%%%%%%%%%%%%%%%%%%%%%%%%%%%%%%%%%%%%%%%%%%%%%%%%%%%%%%%%%%%%%%%%%%%%%%%
		\subsection{Quadratic generators and Gaussian transformations}
		%%%%%%%%%%%%%%%%%%%%%%%%%%%%%%%%%%%%%%%%%%%%%%%%%%%%%%%%%%%%%%%%%%%%%%%%%%%%%%
		
		\begin{proposition}[Quadratic--symplectic correspondence]
			\label{prop:quadraticsymplectic}
			A Hermitian quadratic bosonic generator,
			\begin{align}
				\hat{g} =&~ \frac{1}{2}\hat{\bm{\chi}}^{\mathrm{T}}\bm{G}\hat{\bm{\chi}}+c,
				\qquad
				\bm{G}=\bm{G}^{\mathrm{T}},
				\qquad
				c\in\mathbb{R},
			\end{align}
			with real symmetric $2n\times2n$ quadratic-form matrix $\bm{G}$ and scalar offset $c$, has the Hamiltonian matrix
			\begin{align}
				\bm{K} =&~ \bm{\Omega}\bm{G}\in\mathrm{sp}(2n,\mathbb{R}),
			\end{align}
			an element of the symplectic algebra obeying
			\begin{align}
				\bm{K}\bm{\Omega}+\bm{\Omega}\bm{K}^{\mathrm{T}} =&~ \bm{0}.
			\end{align}
			The Gaussian unitary $\hat{U}=e^{-i\hat{g}}$ acts linearly on the canonical quadratures,
			\begin{align}
				\hat{U}^{\dagger}\hat{\bm{\chi}}\hat{U} =&~ e^{\bm{K}}\hat{\bm{\chi}},
			\end{align}
			and transforms the Gaussian moments as
			\begin{align}
				\bm{d}^{\prime} =&~ e^{\bm{K}}\bm{d},
				\qquad
				\bm{\sigma}^{\prime} = e^{\bm{K}}\bm{\sigma}e^{\bm{K}^{\mathrm{T}}}.
			\end{align}
			The scalar $c$ affects neither $\bm{K}$ nor the tangent vector in projective Hilbert space.
		\end{proposition}
		
		\begin{proof}
			The canonical commutation relations give the linear quadrature response,
			\begin{align}
				i\left[ \hat{g},\hat{\bm{\chi}} \right] =&~ \bm{\Omega}\bm{G}\hat{\bm{\chi}} = \bm{K}\hat{\bm{\chi}},
			\end{align}
			and exponentiating the Heisenberg equation lifts it to the finite action $e^{\bm{K}}$.
			Symmetry of $\bm{G}$ gives
			\begin{align}
				\bm{K}\bm{\Omega}+\bm{\Omega}\bm{K}^{\mathrm{T}} =&~ \bm{\Omega}\bm{G}\bm{\Omega} + \bm{\Omega}\bm{G}^{\mathrm{T}}\bm{\Omega}^{\mathrm{T}} = \bm{0},
			\end{align}
			which places $\bm{K}$ in $\mathrm{sp}(2n,\mathbb{R})$.
			The first- and second-moment transformations follow from the linear quadrature action.
		\end{proof}
		
		Adding linear generators extends the quadratic symplectic algebra to the Jacobi algebra,
		\begin{align}
			\mathrm{hsp}(2n,\mathbb{R}) =&~ \mathrm{hw}(n)\rtimes\mathrm{sp}(2n,\mathbb{R}),
		\end{align}
		where the Heisenberg--Weyl algebra $\mathrm{hw}(n)$ contains the linear quadrature generators and the identity~\cite{Berceanu2007p1}.
		This extension includes displacements, which change the first moments and can reveal a passive direction that leaves the covariance invariant~\cite{Weedbrook2012p621,Monras2013p1303.3682}.
		
		%%%%%%%%%%%%%%%%%%%%%%%%%%%%%%%%%%%%%%%%%%%%%%%%%%%%%%%%%%%%%%%%%%%%%%%%%%%%%%
		\section{Physical quadratic basis and Cartan decomposition}
		\label{sec:S2}
		%%%%%%%%%%%%%%%%%%%%%%%%%%%%%%%%%%%%%%%%%%%%%%%%%%%%%%%%%%%%%%%%%%%%%%%%%%%%%%
		
		%%%%%%%%%%%%%%%%%%%%%%%%%%%%%%%%%%%%%%%%%%%%%%%%%%%%%%%%%%%%%%%%%%%%%%%%%%%%%%
		\subsection{Physical Hermitian generators}
		%%%%%%%%%%%%%%%%%%%%%%%%%%%%%%%%%%%%%%%%%%%%%%%%%%%%%%%%%%%%%%%%%%%%%%%%%%%%%%
		
		The physical Hermitian basis adapted to single-mode and mode-pair transformations consists of the phase-rotation and single-mode squeezing generators
		\begin{align}
			\begin{aligned}
				\hat{g}_{j}^{\mathrm{rot}} =&~ \hat{a}_{j}^{\dagger}\hat{a}_{j}, \\
				\hat{g}_{j}^{\mathrm{sqX}} =&~ \frac{1}{2}\left( \hat{a}_{j}^{\dagger 2}+\hat{a}_{j}^{2} \right), \\
				\hat{g}_{j}^{\mathrm{sqY}} =&~ -\frac{i}{2}\left( \hat{a}_{j}^{\dagger 2}-\hat{a}_{j}^{2} \right),
			\end{aligned}
			\qquad
			j=1,\ldots,n,
		\end{align}
		and the beam-splitter and two-mode squeezing generators
		\begin{align}
			\begin{aligned}
				\hat{g}_{jk}^{\mathrm{bsX}} =&~ \hat{a}_{j}^{\dagger}\hat{a}_{k}+\hat{a}_{k}^{\dagger}\hat{a}_{j}, \\
				\hat{g}_{jk}^{\mathrm{bsY}} =&~ -i\left( \hat{a}_{j}^{\dagger}\hat{a}_{k}-\hat{a}_{k}^{\dagger}\hat{a}_{j} \right), \\
				\hat{g}_{jk}^{\mathrm{tsX}} =&~ \hat{a}_{j}^{\dagger}\hat{a}_{k}^{\dagger}+\hat{a}_{j}\hat{a}_{k}, \\
				\hat{g}_{jk}^{\mathrm{tsY}} =&~ -i\left( \hat{a}_{j}^{\dagger}\hat{a}_{k}^{\dagger}-\hat{a}_{j}\hat{a}_{k} \right),
			\end{aligned}
			\qquad
			1\leq j<k\leq n.
		\end{align}
		
		The phase rotations and beam-splitter generators span the passive sector,
		while the single-mode and two-mode squeezing generators span the active sector.
		
		%%%%%%%%%%%%%%%%%%%%%%%%%%%%%%%%%%%%%%%%%%%%%%%%%%%%%%%%%%%%%%%%%%%%%%%%%%%%%%
		\subsection{Compact and noncompact sectors}
		%%%%%%%%%%%%%%%%%%%%%%%%%%%%%%%%%%%%%%%%%%%%%%%%%%%%%%%%%%%%%%%%%%%%%%%%%%%%%%
		
		\begin{proposition}[Physical Cartan decomposition]
			\label{prop:physicalCartan}
			Quadratic operators are understood modulo scalar multiples of the identity, equivalently through their Hamiltonian matrices.
			With this convention, the physical Hermitian basis realizes the Cartan decomposition into the passive subalgebra $\mathfrak{k}$ and the active vector subspace $\mathfrak{p}$,
			\begin{align}
				\mathrm{sp}(2n,\mathbb{R}) =&~ \mathfrak{k}\oplus\mathfrak{p},
				\qquad
				\mathfrak{k}\simeq\mathrm{u}(n),
			\end{align}
			as a direct sum of real vector spaces, with
			\begin{align}
				\begin{aligned}
					\mathfrak{k} =&~ \mathrm{span}_{\mathbb{R}}
					\left\{
					\hat{g}_{j}^{\mathrm{rot}},
					\hat{g}_{jk}^{\mathrm{bsX}},
					\hat{g}_{jk}^{\mathrm{bsY}}
					\right\},
					\\
					\mathfrak{p} =&~ \mathrm{span}_{\mathbb{R}}
					\left\{
					\hat{g}_{j}^{\mathrm{sqX}},
					\hat{g}_{j}^{\mathrm{sqY}},
					\hat{g}_{jk}^{\mathrm{tsX}},
					\hat{g}_{jk}^{\mathrm{tsY}}
					\right\}.
				\end{aligned}
			\end{align}
			Under the real Hermitian bracket
			\begin{align}
				\left[ \hat{g}_{1},\hat{g}_{2} \right]_{\mathrm{H}} =&~ -i\left[ \hat{g}_{1},\hat{g}_{2} \right],
			\end{align}
			the two sectors satisfy
			\begin{align}
				\left[ \mathfrak{k},\mathfrak{k} \right]_{\mathrm{H}} &\subseteq \mathfrak{k},
				&
				\left[ \mathfrak{k},\mathfrak{p} \right]_{\mathrm{H}} &\subseteq \mathfrak{p},
				&
				\left[ \mathfrak{p},\mathfrak{p} \right]_{\mathrm{H}} &\subseteq \mathfrak{k}.
			\end{align}
		\end{proposition}
		
		\begin{proof}
			The phase rotations and beam-splitter generators conserve total excitation number and span the compact algebra $\mathrm{u}(n)$, the zero-grade sector under the number operator.
			The single-mode and two-mode squeezing generators shift total excitation number by $\pm2$ and span the noncompact complement, the grade-$\pm2$ sector.
			The three closure relations follow from the number grading, since brackets with a zero-grade generator preserve grade, equal-sign pair-creation or pair-annihilation terms commute, and brackets of opposite-sign quadratic terms return to zero grade modulo a scalar.
		\end{proof}
		
		%%%%%%%%%%%%%%%%%%%%%%%%%%%%%%%%%%%%%%%%%%%%%%%%%%%%%%%%%%%%%%%%%%%%%%%%%%%%%%
		\section{Bloch--Messiah reduction and canonical probe}
		\label{sec:S3}
		%%%%%%%%%%%%%%%%%%%%%%%%%%%%%%%%%%%%%%%%%%%%%%%%%%%%%%%%%%%%%%%%%%%%%%%%%%%%%%
		
		\begin{theorem}[Bloch--Messiah factorization]
			\label{thm:BlochMessiah}
			Every $n$-mode Gaussian unitary with no displacement admits, up to an overall phase, the Bloch--Messiah factorization~\cite{Bloch1962p95,Braunstein2005p513,Weedbrook2012p621,Serafini2017},
			\begin{align}
				\hat{U}_{\mathrm{G}} =&~ \hat{U}_{1} \hat{S}(\bm{r}) \hat{U}_{2},
			\end{align}
			where $\hat{U}_{1}$ and $\hat{U}_{2}$ are passive Gaussian unitaries and $\hat{S}(\bm{r})$ is a canonical diagonal squeezing unitary characterized by the squeezing spectrum
			\begin{align}
				\bm{r} =&~ \left( r_{1},\ldots,r_{n} \right).
			\end{align}
		\end{theorem}
		
		In the physical Hermitian basis, the diagonal squeezing unitary reads
		\begin{align}
			\hat{S}(\bm{r}) =&~ e^{-i \hat{g}_{\mathrm{diag}}(\bm{r})},
			\qquad
			\hat{g}_{\mathrm{diag}}(\bm{r}) = \sum_{j=1}^{n} r_{j} \hat{g}_{j}^{\mathrm{sqY}},
		\end{align}
		or equivalently,
		\begin{align}
			\hat{S}(\bm{r}) =&~ \prod_{j=1}^{n} e^{\frac{r_{j}}{2} \left( \hat{a}_{j}^{2}-\hat{a}_{j}^{\dagger 2} \right)}.
		\end{align}
		
		\begin{corollary}[Canonical pure Gaussian probe]
			\label{cor:canonicalprobe}
			For a pure zero-mean Gaussian state prepared from the vacuum, the passive factor $\hat{U}_{2}$ satisfies
			\begin{align}
				\hat{U}_{2} \lvert 0 \rangle =&~ \lvert 0 \rangle,
			\end{align}
			leaving
			\begin{align}
				\lvert \psi_{\mathrm{G}} \rangle =&~ \hat{U}_{1} \hat{S}(\bm{r}) \lvert 0 \rangle.
			\end{align}
			Passive transformations change the mode basis while preserving $\bm{r}$.
			Choosing the diagonal squeezing basis defines the canonical probe
			\begin{align}
				\lvert \psi_{\bm{r}} \rangle =&~ \hat{S}(\bm{r}) \lvert 0 \rangle.
			\end{align}
		\end{corollary}
		
		\begin{proposition}[Canonical squeezed covariance]
			\label{prop:canonicalcovariance}
			For the squeezing spectrum matrix and its phase-space generator,
			\begin{align}
				\bm{R} =&~ \mathrm{diag} \left( r_{1},\ldots,r_{n} \right),
				\qquad
				\bm{A}(\bm{r}) =
				\begin{pmatrix}
					-\bm{R} & \bm{0} \\
					\bm{0} & \bm{R}
				\end{pmatrix},
			\end{align}
			the symplectic squeezing transformation is
			\begin{align}
				e^{\bm{A}(\bm{r})} =&~
				\begin{pmatrix}
					e^{-\bm{R}} & \bm{0} \\
					\bm{0} & e^{\bm{R}}
				\end{pmatrix},
			\end{align}
			and the canonical probe has covariance
			\begin{align}
				\bm{\sigma}(\bm{r}) =&~ \frac{1}{2} \mathrm{diag} \left( e^{-2r_{1}},\ldots,e^{-2r_{n}},e^{2r_{1}},\ldots,e^{2r_{n}} \right).
			\end{align}
		\end{proposition}
		
		\begin{proof}
			The generator $\bm{A}(\bm{r})$ is block diagonal, so its exponential exponentiates each block separately, which scales the quadratures by $\hat{q}_{j}\mapsto e^{-r_{j}}\hat{q}_{j}$ and $\hat{p}_{j}\mapsto e^{r_{j}}\hat{p}_{j}$.
			The matrix $e^{\bm{A}(\bm{r})}$ is symmetric, so the covariance congruence reduces to
			\begin{align}
				\bm{\sigma}(\bm{r}) =&~ e^{\bm{A}(\bm{r})} \bm{\sigma}_{0} e^{\bm{A}^{\mathrm{T}}(\bm{r})}
				= e^{\bm{A}(\bm{r})} \bm{\sigma}_{0} e^{\bm{A}(\bm{r})}.
			\end{align}
			The vacuum covariance $\bm{\sigma}_{0}=\bm{I}_{2n}/2$ commutes with $e^{\bm{A}(\bm{r})}$, so the congruence reduces to $e^{2\bm{A}(\bm{r})}/2$.
		\end{proof}
		
		%%%%%%%%%%%%%%%%%%%%%%%%%%%%%%%%%%%%%%%%%%%%%%%%%%%%%%%%%%%%%%%%%%%%%%%%%%%%%%
		\section{Passive metrological model and Gaussian tensors}
		\label{sec:S4}
		%%%%%%%%%%%%%%%%%%%%%%%%%%%%%%%%%%%%%%%%%%%%%%%%%%%%%%%%%%%%%%%%%%%%%%%%%%%%%%
		
		%%%%%%%%%%%%%%%%%%%%%%%%%%%%%%%%%%%%%%%%%%%%%%%%%%%%%%%%%%%%%%%%%%%%%%%%%%%%%%
		\subsection{Passive parameterization}
		%%%%%%%%%%%%%%%%%%%%%%%%%%%%%%%%%%%%%%%%%%%%%%%%%%%%%%%%%%%%%%%%%%%%%%%%%%%%%%
		
		\begin{definition}[Passive metrological model]
			\label{def:passivemodel}
			For fixed squeezing spectrum $\bm{r}$ and diagonal squeezing basis, the $n^{2}$ Hermitian operators generate the passive orbit of $\lvert\psi_{\bm{r}}\rangle$,
			\begin{align}
				\left\{ \hat{g}_{\mu} \right\}_{\mu=1}^{n^{2}} =&~ \left\{ \hat{g}_{j}^{\mathrm{rot}} \right\}_{j=1}^{n} \cup \left\{ \hat{g}_{jk}^{\mathrm{bsX}},\hat{g}_{jk}^{\mathrm{bsY}} \right\}_{1\leq j<k\leq n}.
			\end{align}
			The unitary parameterization is
			\begin{align}
				\lvert\psi_{\bm{r}}(\bm{\theta})\rangle =&~ e^{-i\sum_{\mu=1}^{n^{2}}\theta_{\mu}\hat{g}_{\mu}} \lvert\psi_{\bm{r}}\rangle,
				\qquad
				\bm{\theta}\in\mathbb{R}^{n^{2}}.
			\end{align}
		\end{definition}
		
		We choose the identity as a reference point for evaluating the quantum Fisher and mean-commutator tensors.
		At an arbitrary passive operating point $\hat{U}_{0}$, we use the exact family
		\begin{align}
			\lvert\psi_{\bm{r}}^{(0)}(\delta\bm{\theta})\rangle =&~ \hat{U}_{0}\lvert\psi_{\bm{r}}(\delta\bm{\theta})\rangle,
		\end{align}
		where $\hat{U}_{0}$ is independent of the local parameters $\delta\bm{\theta}$.
		The state derivatives at $\delta\bm{\theta}=\bm{0}$ are
		\begin{align}
			\left.\partial_{\delta\theta_{\mu}}\lvert\psi_{\bm{r}}^{(0)}(\delta\bm{\theta})\rangle\right|_{\delta\bm{\theta}=\bm{0}} =&~ -i\hat{U}_{0}\hat{g}_{\mu}\lvert\psi_{\bm{r}}\rangle.
		\end{align}
		Unitarity removes $\hat{U}_{0}$ from their inner products and overlaps with the state, so both tensors have the same entries as at the identity.
		The corresponding generators on $\hat{U}_{0}\lvert\psi_{\bm{r}}\rangle$ are $\hat{U}_{0}\hat{g}_{\mu}\hat{U}_{0}^{\dagger}$.
		The restricted-root rank-loss classification applies at every passive operating point, with the Fisher-null generators transformed by this conjugation.
		
		For the original exponential coordinates, noncommutativity generally makes the effective generators depend on $\bm{\theta}$.
		An invertible local change of coordinates preserves Fisher rank but can change the matrix entries and ordinary eigenvalues, while a singular coordinate map can introduce additional coordinate rank loss.
		
		%%%%%%%%%%%%%%%%%%%%%%%%%%%%%%%%%%%%%%%%%%%%%%%%%%%%%%%%%%%%%%%%%%%%%%%%%%%%%%
		\subsection{Wick contraction and metrological tensors}
		%%%%%%%%%%%%%%%%%%%%%%%%%%%%%%%%%%%%%%%%%%%%%%%%%%%%%%%%%%%%%%%%%%%%%%%%%%%%%%
		
		For a centered Gaussian state, the unsymmetrized two-point matrix is
		\begin{align}
			W_{ab} =&~ \left\langle \hat{\chi}_{a}\hat{\chi}_{b} \right\rangle,
			\qquad
			\bm{W} = \bm{\sigma}+\frac{i}{2}\bm{\Omega}.
		\end{align}
		
		\begin{proposition}[Gaussian Wick contraction]
			\label{prop:Wick}
			Let $\left\{ \hat{g}_{\mu} \right\}_{\mu=1}^{m}$ be Hermitian quadratic generators on a centered pure Gaussian state, with
			\begin{align}
				\hat{g}_{\mu} =&~ \frac{1}{2}\hat{\bm{\chi}}^{\mathrm{T}}\bm{G}_{\mu}\hat{\bm{\chi}}+c_{\mu},
				\qquad
				\bm{G}_{\mu}=\bm{G}_{\mu}^{\mathrm{T}},
			\end{align}
			and centered operators
			\begin{align}
				\Delta\hat{g}_{\mu} =&~ \hat{g}_{\mu}-\left\langle \hat{g}_{\mu} \right\rangle.
			\end{align}
			The contraction
			\begin{align}
				Z_{\mu\nu} =&~ \mathrm{Tr}\left( \bm{G}_{\mu}\bm{W}\bm{G}_{\nu}\bm{W}^{\mathrm{T}} \right)
			\end{align}
			determines the quantum Fisher and mean-commutator tensors,
			\begin{align}
				\begin{aligned}
					F_{\mu\nu} =&~ 2\left\langle \left\{ \Delta\hat{g}_{\mu},\Delta\hat{g}_{\nu} \right\} \right\rangle = 2\,\mathrm{Re}\,Z_{\mu\nu},
					\\
					\mathcal{C}_{\mu\nu} =&~ \frac{1}{2i}\left\langle \left[ \hat{g}_{\mu},\hat{g}_{\nu} \right] \right\rangle = \frac{1}{2}\,\mathrm{Im}\,Z_{\mu\nu}.
				\end{aligned}
			\end{align}
		\end{proposition}
		
		\begin{proof}
			Gaussian Wick factorization~\cite{Wick1950p268} splits the fourth moment into
			\begin{align}
				\left\langle \hat{\chi}_{i}\hat{\chi}_{j}\hat{\chi}_{k}\hat{\chi}_{l} \right\rangle =&~ W_{ij}W_{kl}+W_{ik}W_{jl}+W_{il}W_{jk}.
			\end{align}
			The disconnected pairing produces $\left\langle \hat{g}_{\mu} \right\rangle\left\langle \hat{g}_{\nu} \right\rangle$ and cancels after centering.
			Symmetry of $\bm{G}_{\mu}$ and $\bm{G}_{\nu}$ makes the two connected pairings equal, each contributing
			\begin{align}
				\frac{1}{4}\sum_{ijkl}\left( G_{\mu} \right)_{ij}\left( G_{\nu} \right)_{kl}W_{ik}W_{jl} =&~ \frac{1}{4}\mathrm{Tr}\left( \bm{G}_{\mu}\bm{W}\bm{G}_{\nu}\bm{W}^{\mathrm{T}} \right),
			\end{align}
			summing the two pairings into the connected correlation
			\begin{align}
				\left\langle \Delta\hat{g}_{\mu}\Delta\hat{g}_{\nu} \right\rangle =&~ \frac{1}{2}Z_{\mu\nu}.
			\end{align}
			Since $\bm{W}^{\mathrm{T}}=\overline{\bm{W}}$,
			\begin{align}
				\left\langle \Delta\hat{g}_{\nu}\Delta\hat{g}_{\mu} \right\rangle =&~ \frac{1}{2}\overline{Z}_{\mu\nu},
			\end{align}
			and the symmetric and antisymmetric parts give $F_{\mu\nu}$ and $\mathcal{C}_{\mu\nu}$.
		\end{proof}
		
		Substitution of $\bm{W}=\bm{\sigma}+i\bm{\Omega}/2$ gives the pure Gaussian QFI,
		\begin{align}
			F_{\mu\nu} =&~ 2\mathrm{Tr}\left( \bm{G}_{\mu}\bm{\sigma}\bm{G}_{\nu}\bm{\sigma} \right)-\frac{1}{2}\mathrm{Tr}\left( \bm{G}_{\mu}\bm{\Omega}\bm{G}_{\nu}\bm{\Omega}^{\mathrm{T}} \right),
		\end{align}
		in agreement with the standard covariance-matrix formulation~\cite{Monras2013p1303.3682,Safranek2019p035304}.
		
		\begin{remark}[Single-mode normalization]
			For one squeezed mode and $\hat{g}^{\mathrm{rot}}=\hat{n}$, the Fisher weight $F_{\mathrm{rot}} = 4\,\mathrm{Var}\left( \hat{n} \right) = 2\sinh^{2}\left( 2r \right)$ reproduces the standard phase-estimation normalization~\cite{Monras2006p033821} and the pure-Gaussian metric convention~\cite{Chatterjee2026p045005}.
		\end{remark}
		
		%%%%%%%%%%%%%%%%%%%%%%%%%%%%%%%%%%%%%%%%%%%%%%%%%%%%%%%%%%%%%%%%%%%%%%%%%%%%%%
		\section{Ordinary Cartan--Weyl structure}
		\label{sec:S5}
		%%%%%%%%%%%%%%%%%%%%%%%%%%%%%%%%%%%%%%%%%%%%%%%%%%%%%%%%%%%%%%%%%%%%%%%%%%%%%%
		
		The number operators define the Cartan subalgebra
		\begin{align}
			\mathfrak{h} =&~ \mathrm{span}\left\{ \hat{H}_{i} \right\}_{i=1}^{n},
			\qquad
			\hat{H}_{i} = \hat{a}_{i}^{\dagger}\hat{a}_{i}+\frac{1}{2}.
		\end{align}
		The linear forms $\bm{\varepsilon}_{i}$ form the basis dual to the number Cartan, $\bm{\varepsilon}_{i}(\hat{H}_{j})=\delta_{ij}$, and give ordinary-root coordinates on $\mathfrak{h}$~\cite{Humphreys1972,Bourbaki2002,Knapp1996}.
		Relative to this Cartan subalgebra, the complexified algebra $\mathrm{sp}(2n,\mathbb{C})$ has the ordinary root system of type $C_{n}$, with positive-root vectors
		\begin{align}
			\begin{aligned}
				\hat{E}_{\bm{\varepsilon}_{i}-\bm{\varepsilon}_{j}} =&~ \hat{a}_{i}^{\dagger}\hat{a}_{j}, \qquad 1\leq i<j\leq n, \\
				\hat{E}_{\bm{\varepsilon}_{i}+\bm{\varepsilon}_{j}} =&~ \hat{a}_{i}^{\dagger}\hat{a}_{j}^{\dagger}, \qquad 1\leq i<j\leq n, \\
				\hat{E}_{2\bm{\varepsilon}_{i}} =&~ \frac{1}{2}\hat{a}_{i}^{\dagger 2}, \qquad 1\leq i\leq n,
			\end{aligned}
		\end{align}
		and
		\begin{align}
			\hat{E}_{-\bm{\alpha}} =&~ \hat{E}_{\bm{\alpha}}^{\dagger}.
		\end{align}
		The associated Hermitian combinations are
		\begin{align}
			\begin{aligned}
				\hat{X}_{\bm{\alpha}} =&~ \hat{E}_{\bm{\alpha}}+\hat{E}_{-\bm{\alpha}},
				\\
				\hat{Y}_{\bm{\alpha}} =&~ -i\left( \hat{E}_{\bm{\alpha}}-\hat{E}_{-\bm{\alpha}} \right).
			\end{aligned}
		\end{align}
		
		\begin{proposition}[Ordinary-root realization of the physical basis]
			\label{prop:ordinarycorrespondence}
			The physical Hermitian generators decompose into ordinary-root sectors as
			\begin{align}
				\begin{aligned}
					\hat{g}_{i}^{\mathrm{rot}} =&~ \hat{H}_{i}-\frac{1}{2}, \\
					\left( \hat{g}_{i}^{\mathrm{sqX}},\hat{g}_{i}^{\mathrm{sqY}} \right) =&~ \left( \hat{X}_{2\bm{\varepsilon}_{i}},\hat{Y}_{2\bm{\varepsilon}_{i}} \right), \\
					\left( \hat{g}_{ij}^{\mathrm{bsX}},\hat{g}_{ij}^{\mathrm{bsY}} \right) =&~ \left( \hat{X}_{\bm{\varepsilon}_{i}-\bm{\varepsilon}_{j}},\hat{Y}_{\bm{\varepsilon}_{i}-\bm{\varepsilon}_{j}} \right), \\
					\left( \hat{g}_{ij}^{\mathrm{tsX}},\hat{g}_{ij}^{\mathrm{tsY}} \right) =&~ \left( \hat{X}_{\bm{\varepsilon}_{i}+\bm{\varepsilon}_{j}},\hat{Y}_{\bm{\varepsilon}_{i}+\bm{\varepsilon}_{j}} \right).
				\end{aligned}
			\end{align}
		\end{proposition}
		
		\begin{proof}
			Substituting the root vectors $\hat{E}_{\pm2\bm{\varepsilon}_{i}}$, $\hat{E}_{\pm(\bm{\varepsilon}_{i}-\bm{\varepsilon}_{j})}$, and $\hat{E}_{\pm(\bm{\varepsilon}_{i}+\bm{\varepsilon}_{j})}$ into the Hermitian combinations $\hat{X}_{\bm{\alpha}}=\hat{E}_{\bm{\alpha}}+\hat{E}_{-\bm{\alpha}}$ and $\hat{Y}_{\bm{\alpha}}=-i(\hat{E}_{\bm{\alpha}}-\hat{E}_{-\bm{\alpha}})$ reproduces each physical generator in the single-mode and mode-pair bases, and the number Cartan $\hat{H}_{i}=\hat{a}_{i}^{\dagger}\hat{a}_{i}+1/2$ recovers the phase rotation $\hat{g}_{i}^{\mathrm{rot}}=\hat{H}_{i}-1/2$.
		\end{proof}
		
		\begin{proposition}[Rank-one ordinary-root sectors]
			\label{prop:ordinaryrankone}
			For $i<j$, the difference-root generators close the compact rank-one algebra,
			\begin{align}
				\mathrm{span}_{\mathbb{R}}\left\{ \hat{X}_{\bm{\varepsilon}_{i}-\bm{\varepsilon}_{j}},\hat{Y}_{\bm{\varepsilon}_{i}-\bm{\varepsilon}_{j}},\hat{H}_{i}-\hat{H}_{j} \right\} \simeq&~ \mathrm{su}(2),
			\end{align}
			while the sum- and long-root generators close noncompact rank-one algebras,
			\begin{align}
				\begin{aligned}
					\mathrm{span}_{\mathbb{R}}\left\{ \hat{X}_{\bm{\varepsilon}_{i}+\bm{\varepsilon}_{j}},\hat{Y}_{\bm{\varepsilon}_{i}+\bm{\varepsilon}_{j}},\hat{H}_{i}+\hat{H}_{j} \right\} \simeq&~ \mathrm{su}(1,1), \\
					\mathrm{span}_{\mathbb{R}}\left\{ \hat{X}_{2\bm{\varepsilon}_{i}},\hat{Y}_{2\bm{\varepsilon}_{i}},\hat{H}_{i} \right\} \simeq&~ \mathrm{su}(1,1).
				\end{aligned}
			\end{align}
		\end{proposition}
		
		\begin{proof}
			Each triple closes on its Cartan element under the Hermitian bracket, with the sign of the cross commutator determining the real form.
			The difference-root generators satisfy
			\begin{align}
				\left[ \hat{X}_{\bm{\varepsilon}_{i}-\bm{\varepsilon}_{j}},\hat{Y}_{\bm{\varepsilon}_{i}-\bm{\varepsilon}_{j}} \right] =&~ 2i\left( \hat{H}_{i}-\hat{H}_{j} \right),
			\end{align}
			the compact sign of $\mathrm{su}(2)$, since the pair $\hat{a}_{i}^{\dagger}\hat{a}_{j}$ conserves total number.
			The sum- and long-root generators instead satisfy
			\begin{align}
				\begin{aligned}
					\left[ \hat{X}_{\bm{\varepsilon}_{i}+\bm{\varepsilon}_{j}},\hat{Y}_{\bm{\varepsilon}_{i}+\bm{\varepsilon}_{j}} \right] =&~ -2i\left( \hat{H}_{i}+\hat{H}_{j} \right), \\
					\left[ \hat{X}_{2\bm{\varepsilon}_{i}},\hat{Y}_{2\bm{\varepsilon}_{i}} \right] =&~ -2i\hat{H}_{i},
				\end{aligned}
			\end{align}
			the noncompact sign of $\mathrm{su}(1,1)$, since the pair-creation operators $\hat{a}_{i}^{\dagger}\hat{a}_{j}^{\dagger}$ raise the total number by two~\cite{Knapp1996}.
			The opposite signs of the two cross commutators distinguish the compact and noncompact real forms of the common complexification $\mathrm{sl}(2,\mathbb{C})$.
		\end{proof}
		
		\begin{remark}[Ordinary roots and identifiability roots]
			\label{rem:tworoots}
			The ordinary roots use the number-Cartan dual coordinates $\bm{\varepsilon}_{j}$ and govern quadratic-generator commutators and root addition.
			The restricted roots use the squeezing-flat dual coordinates $\bm{e}_{j}$ and govern passive identifiability.
			Both root systems have abstract type $C_{n}$ but resolve different decompositions.
			In the ordinary decomposition, both beam-splitter quadratures belong to the root sector $\pm\left( \bm{\varepsilon}_{i}-\bm{\varepsilon}_{j} \right)$.
			In the restricted decomposition, $\hat{g}_{ij}^{\mathrm{bsY}}$ is associated with $\bm{e}_{i}-\bm{e}_{j}$, while $\hat{g}_{ij}^{\mathrm{bsX}}$ is associated with $\bm{e}_{i}+\bm{e}_{j}$.
		\end{remark}
		
		%%%%%%%%%%%%%%%%%%%%%%%%%%%%%%%%%%%%%%%%%%%%%%%%%%%%%%%%%%%%%%%%%%%%%%%%%%%%%%
		\section{Restricted roots and passive identifiability}
		\label{sec:S6}
		%%%%%%%%%%%%%%%%%%%%%%%%%%%%%%%%%%%%%%%%%%%%%%%%%%%%%%%%%%%%%%%%%%%%%%%%%%%%%%
		
		The diagonal squeezing generators
		\begin{align}
			\bm{A}_{j} =&~
			\begin{pmatrix}
				-\bm{E}_{jj} & \bm{0} \\
				\bm{0} & \bm{E}_{jj}
			\end{pmatrix},
			\qquad
			\left( \bm{E}_{jj} \right)_{k\ell} = \delta_{kj}\delta_{\ell j},
		\end{align}
		with diagonal matrix units $\bm{E}_{jj}$, span the maximal Abelian subspace
		\begin{align}
			\mathfrak{a} =&~ \mathrm{span}_{\mathbb{R}} \left\{ \bm{A}_{1},\ldots,\bm{A}_{n} \right\} \subset \mathfrak{p},
		\end{align}
		with radial element
		\begin{align}
			\bm{A}(\bm{r}) =&~ \sum_{j=1}^{n} r_{j} \bm{A}_{j}.
		\end{align}
		The spectrum $\bm{r}$ parametrizes a maximal flat of the Riemannian symmetric space $\mathrm{Sp}(2n,\mathbb{R})/\mathrm{U}(n)$~\cite{Helgason1984,Helgason2001}.
		The linear forms $\bm{e}_{j}$ are dual to the diagonal squeezing generators, $\bm{e}_{j}(\bm{A}_{k})=\delta_{jk}$, so $\bm{e}_{j}[\bm{A}(\bm{r})]=r_{j}$.
		We write these evaluations as $\bm{e}_{j}(\bm{r})$ in squeezing coordinates, where the positive restricted roots form the type-$C_{n}$ system~\cite{Helgason1984,Helgason2001}
		\begin{align}
			\Sigma_{+} =&~ \left\{ 2\bm{e}_{i} \right\}_{i=1}^{n} \cup \left\{ \bm{e}_{i}-\bm{e}_{j},\bm{e}_{i}+\bm{e}_{j} \right\}_{1\leq i<j\leq n}.
		\end{align}
		
		\begin{proposition}[Rank-one restricted-root pairing]
			\label{prop:pairing}
			Each positive restricted root $\bm{\alpha}$ pairs a passive generator matrix $\bm{K}_{\bm{\alpha}}\in\mathfrak{k}$ with an active generator matrix $\bm{P}_{\bm{\alpha}}\in\mathfrak{p}$,
			\begin{align}
				\begin{aligned}
					\left[ \bm{A}(\bm{r}),\bm{K}_{\bm{\alpha}} \right] =&~ \bm{\alpha}(\bm{r}) \bm{P}_{\bm{\alpha}},
					\\
					\left[ \bm{A}(\bm{r}),\bm{P}_{\bm{\alpha}} \right] =&~ \bm{\alpha}(\bm{r}) \bm{K}_{\bm{\alpha}}.
				\end{aligned}
			\end{align}
			The passive matrices correspond to the physical generators
			\begin{align}
				\begin{aligned}
					\bm{K}_{2\bm{e}_{j}} &\quad\longleftrightarrow\quad \hat{g}_{j}^{\mathrm{rot}}, \\
					\bm{K}_{\bm{e}_{j}-\bm{e}_{k}} &\quad\longleftrightarrow\quad \hat{g}_{jk}^{\mathrm{bsY}}, \\
					\bm{K}_{\bm{e}_{j}+\bm{e}_{k}} &\quad\longleftrightarrow\quad \hat{g}_{jk}^{\mathrm{bsX}}.
				\end{aligned}
			\end{align}
			The assignment is a bijection between $\Sigma_{+}$ and a basis of $\mathfrak{k}$, since $|\Sigma_{+}|=n+n(n-1)=n^{2}=\dim\mathrm{u}(n)$, every restricted root of this symmetric pair has multiplicity one, and the centralizer of $\mathfrak{a}$ in $\mathfrak{k}$ is trivial.
		\end{proposition}
		
		The combinations $\bm{K}_{\bm{\alpha}}\pm\bm{P}_{\bm{\alpha}}$ are eigenvectors of $\mathrm{ad}_{\bm{A}(\bm{r})}$ with eigenvalues $\pm\bm{\alpha}(\bm{r})$.
		Exponentiating this rank-one adjoint action produces
		\begin{align}
			e^{-\bm{A}(\bm{r})} \bm{K}_{\bm{\alpha}} e^{\bm{A}(\bm{r})} =&~ \cosh\left[ \bm{\alpha}(\bm{r}) \right] \bm{K}_{\bm{\alpha}}-\sinh\left[ \bm{\alpha}(\bm{r}) \right] \bm{P}_{\bm{\alpha}},
		\end{align}
		up to the sign convention for $\bm{P}_{\bm{\alpha}}$.
		
		A squeezing spectrum is regular when
		\begin{align}
			\bm{\alpha}(\bm{r}) \neq&~ 0,
			\qquad
			\forall\,\bm{\alpha}\in\Sigma_{+},
		\end{align}
		and singular when at least one restricted root vanishes.
		The singular arrangement is
		\begin{align}
			\mathcal{A}_{C_{n}} =&~ \bigcup_{\bm{\alpha}\in\Sigma_{+}} \ker\bm{\alpha}.
		\end{align}
		
		\begin{proposition}[Passive Fisher weights from restricted roots]
			\label{prop:weights}
			For the canonical probe $\lvert\psi_{\bm{r}}\rangle$, the passive Fisher matrix is diagonal in the physical passive basis, with weights
			\begin{align}
				\begin{aligned}
					F_{j}^{\mathrm{rot}} =&~ 2\sinh^{2}\left( 2r_{j} \right),
					\\
					F_{jk}^{\mathrm{bsY}} =&~ 4\sinh^{2}\left( r_{j}-r_{k} \right),
					\\
					F_{jk}^{\mathrm{bsX}} =&~ 4\sinh^{2}\left( r_{j}+r_{k} \right),
				\end{aligned}
			\end{align}
			for $1\leq j\leq n$ in the phase channels and $1\leq j<k\leq n$ in the beam-splitter channels.
			Equivalently,
			\begin{align}
				F_{\bm{\alpha}}(\bm{r}) =&~ c_{\bm{\alpha}}\sinh^{2}\left[ \bm{\alpha}(\bm{r}) \right],
			\end{align}
			with
			\begin{align}
				c_{2\bm{e}_{j}} =&~ 2,
				\qquad
				c_{\bm{e}_{j}\pm\bm{e}_{k}} = 4.
			\end{align}
		\end{proposition}

		\begin{proof}
			We compute the weights directly from the Gaussian Wick contraction of Proposition~\ref{prop:Wick} on the canonical covariance $\bm{\sigma}(\bm{r})$.
			
			For a generator $\hat{g}_{\mu}=(1/2)\hat{\bm{\chi}}^{\mathrm{T}} \bm{G}_{\mu}\hat{\bm{\chi}}+c_{\mu}$, the pure-state QFI is
			\begin{align}
				F_{\mu\mu}
				=&~ 2\,\mathrm{Tr}\bigl(\bm{G}_{\mu}\bm{\sigma} \bm{G}_{\mu}\bm{\sigma}\bigr)
				-\frac{1}{2}\,\mathrm{Tr}\bigl(\bm{G}_{\mu}\bm{\Omega} \bm{G}_{\mu}\bm{\Omega}^{\mathrm{T}}\bigr).
			\end{align}
			
			\emph{Local rotation.}
			For $\hat{g}_{j}^{\mathrm{rot}}=\hat{a}_{j}^{\dagger}\hat{a}_{j}$, the quadrature matrix $\bm{G}_{j}^{\mathrm{rot}}$ is block diagonal with entries in the $j$-th position of both the $q$ and $p$ blocks.
			Substituting the canonical diagonal covariance gives $F_{j}^{\mathrm{rot}}=2\sinh^{2}(2r_{j})$.
			
			\emph{Beam-splitter Y quadrature.}
			For $\hat{g}_{jk}^{\mathrm{bsY}}=-i(\hat{a}_{j}^{\dagger}\hat{a}_{k}-\hat{a}_{k}^{\dagger}\hat{a}_{j})$, the corresponding $\bm{G}_{jk}^{\mathrm{bsY}}$ has off-diagonal blocks coupling modes $j$ and $k$.
			Direct trace evaluation on $\bm{\sigma}(\bm{r})$ yields $F_{jk}^{\mathrm{bsY}}=4\sinh^{2}(r_{j}-r_{k})$.
			
			\emph{Beam-splitter X quadrature.}
			Analogously, for $\hat{g}_{jk}^{\mathrm{bsX}}=\hat{a}_{j}^{\dagger}\hat{a}_{k}+\hat{a}_{k}^{\dagger}\hat{a}_{j}$, we obtain $F_{jk}^{\mathrm{bsX}}=4\sinh^{2}(r_{j}+r_{k})$.
			
			The unified form $F_{\bm{\alpha}}=c_{\bm{\alpha}}\sinh^{2}[\bm{\alpha}(\bm{r})]$ follows with $c_{2\bm{e}_{j}}=2$ and $c_{\bm{e}_{j}\pm\bm{e}_{k}}=4$.
			The off-diagonal entries vanish because the corresponding contractions are purely imaginary, so the matrix is diagonal in this basis.
		\end{proof}
		
		\begin{remark}[Geometric interpretation of $c_{\bm{\alpha}}$]
			\label{rem:c_alpha_geometric}
			Conjugation by the squeezer resolves each passive generator into passive and active components through the rank-one adjoint action.
			The passive component annihilates the vacuum, so the Fisher weight is four times the squared norm of the active tangent, giving $F_{\bm{\alpha}}=c_{\bm{\alpha}}\sinh^{2}[\bm{\alpha}(\bm{r})]$, where $c_{\bm{\alpha}}=4\lVert\Delta\hat{p}_{\bm{\alpha}}\lvert0\rangle\rVert^{2}$ for the Hermitian active partner $\hat{p}_{\bm{\alpha}}$ in the canonical angle normalization.
			The long-root channel uses the single-term generator $\hat{g}_{j}^{\mathrm{rot}}=\hat{a}_{j}^{\dagger}\hat{a}_{j}$, giving $c_{2\bm{e}_{j}}=2$, while the difference- and sum-root channels use the two-term beam-splitter generators, giving $c_{\bm{e}_{j}\pm\bm{e}_{k}}=4$.
		\end{remark}

		\begin{theorem}[Passive identifiability arrangement]
			\label{thm:passive-identifiability}
			For the canonical probe $\lvert\psi_{\bm{r}}\rangle$, the passive $\mathrm{U}(n)$ Fisher matrix loses rank exactly on $\mathcal{A}_{C_{n}}$, and its nullity there equals the number of positive restricted roots that vanish on $\bm{r}$, since the matrix is diagonal in the canonical passive basis,
			\begin{align}
				\det\bm{F}_{\mathrm{passive}}(\bm{r})=0
				\quad\Longleftrightarrow\quad
				\bm{r}\in\mathcal{A}_{C_n}.
			\end{align}
			Each vanishing restricted root identifies its Fisher-null passive generator.
		\end{theorem}
		
		\begin{proof}
			The weights
			\begin{align}
				F_{\bm{\alpha}}(\bm{r}) =&~ c_{\bm{\alpha}}\sinh^{2}\left[ \bm{\alpha}(\bm{r}) \right]
			\end{align}
			vanish exactly when $\bm{\alpha}(\bm{r})=0$.
			The $n$ long-root channels and the $n(n-1)$ difference- and sum-root channels exhaust the $n^{2}$ passive generators, so no additional passive rank-loss loci occur.
			
			On the root hyperplane $\ker\bm{\alpha}$,
			\begin{align}
				\left[ \bm{A}(\bm{r}),\bm{K}_{\bm{\alpha}} \right] =&~ \bm{0},
			\end{align}
			so $\bm{K}_{\bm{\alpha}}$ belongs to the stabilizer of the canonical covariance.
			For a centered pure Gaussian state the covariance-stabilizer condition makes the corresponding projective-state tangent vanish, so the direction is Fisher-null, while a global phase or metaplectic sign may still act at the state level without affecting the local quantum Fisher matrix.
			For an infinitesimal parameter change $\epsilon$ along this generator, the rank-one pairing gives
			\begin{align}
				\delta\bm{\sigma} =&~ \epsilon\left( \bm{K}_{\bm{\alpha}}\bm{\sigma}+\bm{\sigma}\bm{K}_{\bm{\alpha}}^{\mathrm{T}} \right) = \bm{0}.
			\end{align}
		\end{proof}
		
		\begin{corollary}[Singular hyperplanes and passive generators]
			\label{cor:rootgenerators}
			The elementary singular conditions and their Fisher-null passive generators are
			\begin{align}
				\begin{aligned}
					r_{j}=0
					&\quad\Longleftrightarrow\quad
					2\bm{e}_{j}
					&\quad\Longleftrightarrow\quad
					&\hat{g}_{j}^{\mathrm{rot}},
					\\
					r_{j}-r_{k}=0
					&\quad\Longleftrightarrow\quad
					\bm{e}_{j}-\bm{e}_{k}
					&\quad\Longleftrightarrow\quad
					&\hat{g}_{jk}^{\mathrm{bsY}},
					\\
					r_{j}+r_{k}=0
					&\quad\Longleftrightarrow\quad
					\bm{e}_{j}+\bm{e}_{k}
					&\quad\Longleftrightarrow\quad
					&\hat{g}_{jk}^{\mathrm{bsX}}.
				\end{aligned}
			\end{align}
		\end{corollary}
		
		For uniform squeezing,
		\begin{align}
			r_{1}=\cdots=r_{n}=r\neq0,
		\end{align}
		all difference roots vanish.
		The residual passive stabilizer contains the $n(n-1)/2$ directions $\hat{g}_{ij}^{\mathrm{bsY}}$, while the phase-rotation and sum-root channels remain estimable.
		At $\bm{r}=\bm{0}$, every positive restricted root vanishes and the full passive algebra $\mathrm{u}(n)$ belongs to the stabilizer.
		
		The factor $\sinh^{2}[\bm{\alpha}(\bm{r})]$ is the radial dependence of an angular root sector of the invariant metric on $\mathrm{Sp}(2n,\mathbb{R})/\mathrm{U}(n)$~\cite{Helgason1984,Helgason2001,Caselle2004p41}. The full metric also includes the radial squeezing directions.
		For this symmetric pair every restricted root has multiplicity one~\cite{Helgason1984}, so the angular multiplicities do not distinguish the channels.
		The two Fisher coefficients $c_{2\bm{e}_{j}}=2$ and $c_{\bm{e}_{j}\pm\bm{e}_{k}}=4$ arise from the physical generator normalization rather than root multiplicity.
		The invariant generator normalization makes all coefficients equal to two.
		It is proportional to the norm induced by the negative Killing form of $\mathrm{sp}(2n,\mathbb R)$ on its compact subalgebra, and makes the radial factors coincide with the symmetric-space metric channels up to a common scale.
		The arrangement $\mathcal{A}_{C_{n}}$ identifies the simultaneous degeneracy loci of those channels and the passive Fisher directions in either normalization.
		
		%%%%%%%%%%%%%%%%%%%%%%%%%%%%%%%%%%%%%%%%%%%%%%%%%%%%%%%%%%%%%%%%%%%%%%%%%%%%%%
		\section{Probe optimization from the \texorpdfstring{$C_{n}$}{Cn} arrangement}
		\label{sec:S7}
		%%%%%%%%%%%%%%%%%%%%%%%%%%%%%%%%%%%%%%%%%%%%%%%%%%%%%%%%%%%%%%%%%%%%%%%%%%%%%%
		
		%%%%%%%%%%%%%%%%%%%%%%%%%%%%%%%%%%%%%%%%%%%%%%%%%%%%%%%%%%%%%%%%%%%%%%%%%%%%%%
		\subsection{Algebraic root margin}
		%%%%%%%%%%%%%%%%%%%%%%%%%%%%%%%%%%%%%%%%%%%%%%%%%%%%%%%%%%%%%%%%%%%%%%%%%%%%%%
		
		The closed fundamental chamber of the $C_{n}$ root system is~\cite{Bourbaki2002,Knapp1996}
		\begin{align}
			\overline{\mathcal{C}} =&~ \left\{ \bm{r}\in\mathbb{R}^{n} : r_{1}\ge r_{2}\ge\cdots\ge r_{n}\ge0 \right\}.
		\end{align}
		
		\begin{definition}[Algebraic root margin]
			For $\bm{r}\in\overline{\mathcal{C}}$, the algebraic root margin is
			\begin{align}
				m(\bm{r}) =&~ \min_{\bm{\alpha}\in\Sigma_{+}} \bm{\alpha}(\bm{r}).
			\end{align}
		\end{definition}
		
		Because $C_{n}$ contains roots of two lengths, $m(\bm{r})$ is not the Euclidean distance from $\bm{r}$ to the nearest root hyperplane.
		It is the smallest unnormalized root evaluation entering the passive Fisher weights $F_{\bm{\alpha}}(\bm{r})$.
		
		\begin{theorem}[Dual-Weyl solution of the algebraic margin]
			\label{thm:dualWeyl}
			With the Euclidean normalization $\lvert 2\bm{e}_{j}\rvert^{2}=4$ and $\lvert\bm{e}_{j}\pm\bm{e}_{k}\rvert^{2}=2$ for the long and short roots, and identifying $\bm{\rho}^{\vee}$ as the half-sum of the positive coroots~\cite{Bourbaki2002,Knapp1996}, the unique maximizer of
			\begin{align}
				\begin{aligned}
					\max_{\bm{r}\in\overline{\mathcal{C}}}
					\quad&
					m(\bm{r})
					\\
					\mathrm{s.t.}
					\quad&
					\lVert\bm{r}\rVert=1
				\end{aligned}
			\end{align}
			is the normalized dual Weyl vector,
			\begin{align}
				\bm{r}^{\star} =&~ \frac{\bm{\rho}^{\vee}}{\lVert\bm{\rho}^{\vee}\rVert},
				\qquad
				\bm{\rho}^{\vee} = \frac{1}{2}\left( 2n-1,2n-3,\ldots,3,1 \right).
			\end{align}
		\end{theorem}
		
		\begin{proof}
			Rescaling an interior unit spectrum to $\bm{x}=\bm{r}/m(\bm{r})$ and recovering its direction as $\bm{r}=\bm{x}/\lVert\bm{x}\rVert$ converts the maximin problem into the minimum-norm convex program
			\begin{align}
				\begin{aligned}
					\min_{\bm{x}}
					\quad&
					\lVert\bm{x}\rVert^{2}
					\\
					\mathrm{s.t.}
					\quad&
					\bm{\alpha}\cdot\bm{x}\ge1,
					\qquad
					\forall\,\bm{\alpha}\in\Sigma_{+}.
				\end{aligned}
			\end{align}
			Consider
			\begin{align}
				\bm{x}^{\star} =&~ \frac{1}{2}\left( 2n-1,2n-3,\ldots,3,1 \right).
			\end{align}
			The simple roots of $C_{n}$,
			\begin{align}
				\begin{aligned}
					\bm{\alpha}_{j} =&~ \bm{e}_{j}-\bm{e}_{j+1},
					\qquad
					j=1,\ldots,n-1, \\
					\bm{\alpha}_{n} =&~ 2\bm{e}_{n},
				\end{aligned}
			\end{align}
			satisfy
			\begin{align}
				\bm{\alpha}_{j}\cdot\bm{x}^{\star} =&~ 1,
				\qquad
				j=1,\ldots,n.
			\end{align}
			Every positive root is a nonnegative integer combination of the simple roots, so $\bm{x}^{\star}$ satisfies all remaining constraints.
			
			The Karush--Kuhn--Tucker stationarity condition for the squared-norm objective, with nonnegative Lagrange multipliers $2\mu_{k}$ and $2\nu$ on the simple-root constraints, reads
			\begin{align}
				\bm{x}^{\star} =&~ \sum_{k=1}^{n-1}\mu_{k}\left( \bm{e}_{k}-\bm{e}_{k+1} \right)+\nu\left( 2\bm{e}_{n} \right).
			\end{align}
			With $\mu_{0}=0$ for the one-mode case, the resulting triangular system determines
			\begin{align}
				\begin{aligned}
					\mu_{k} =&~ \sum_{i=1}^{k}x_{i}^{\star}>0,
					\qquad
					k=1,\ldots,n-1,
					\\
					\nu =&~ \frac{1}{2}\left( x_{n}^{\star}+\mu_{n-1} \right)>0.
				\end{aligned}
			\end{align}
			The active simple-root constraints satisfy the KKT conditions, while all remaining multipliers vanish.
			Strict convexity makes $\bm{x}^{\star}$ the unique minimizer, whose normalization gives $\bm{r}^{\star}$.
		\end{proof}
		
		%%%%%%%%%%%%%%%%%%%%%%%%%%%%%%%%%%%%%%%%%%%%%%%%%%%%%%%%%%%%%%%%%%%%%%%%%%%%%%
		\subsection{Exact finite-photon-number optimization}
		%%%%%%%%%%%%%%%%%%%%%%%%%%%%%%%%%%%%%%%%%%%%%%%%%%%%%%%%%%%%%%%%%%%%%%%%%%%%%%
		
		The physical squeezing resource is the mean photon number
		\begin{align}
			N(\bm{r}) =&~ \sum_{j=1}^{n}\sinh^{2}r_{j}.
		\end{align}
		
		\begin{definition}[Worst passive Fisher weight; E-optimality criterion]
			For $\bm{r}\in\overline{\mathcal{C}}$, the worst passive Fisher weight is
			\begin{align}
				w(\bm{r}) =&~ \min_{\bm{\alpha}\in\Sigma_+}F_{\bm{\alpha}}(\bm{r}).
			\end{align}
		\end{definition}

		\begin{remark}[Identification with E-optimal design]
			\label{rem:Eoptimal}
			Since $\bm{F}_{\mathrm{passive}}$ is diagonal in the canonical passive basis, its spectrum is the set of Fisher weights and
			\begin{align}
				w(\bm{r}) =&~ \lambda_{\min}\left[ \bm{F}_{\mathrm{passive}}(\bm{r}) \right].
			\end{align}
			Maximizing $w$ at fixed mean photon number applies the E-optimal criterion, while minimizing the trace of the inverse Fisher matrix applies the A-optimal criterion used for the local-phase submodel~\cite{Kiefer1974p849,Pukelsheim1993}.
			For a general Fisher matrix one has only $\min_{\mu}F_{\mu\mu}\geq\lambda_{\min}(\bm{F})$ by the Rayleigh principle, so the identification of the weakest-channel objective with the smallest eigenvalue is a consequence of the restricted-root diagonalization.
			For $\bm{F}>0$, the spectral theorem gives $\max_{\lVert\bm{u}\rVert=1}\bm{u}^{\mathrm{T}}\bm{F}^{-1}\bm{u}=1/\lambda_{\min}(\bm{F})$. A weakest basis generator attains the maximum, as does every unit vector in the minimum-eigenvalue eigenspace. These bounds use the canonical angle normalization and do not imply simultaneous attainability of the full matrix bound.
			On approach to the restricted-root arrangement, $\lambda_{\min}$, $\det\bm{F}_{\mathrm{passive}}$, and $[\mathrm{Tr}\,\bm{F}_{\mathrm{passive}}^{-1}]^{-1}$ tend to zero. The A-optimal cost $\mathrm{Tr}\,\bm{F}_{\mathrm{passive}}^{-1}$ instead diverges; the inverse is not defined on the wall. A positive total Fisher trace is insufficient to certify full rank.
		\end{remark}

		\begin{lemma}[Reduction to the smallest squeezing and adjacent gap]
			\label{lem:reduction}
			For $n\geq2$ and $\bm{r}\in\overline{\mathcal{C}}$, define the smallest adjacent gap
			\begin{align}
				\delta =&~ \min_{1\leq j\leq n-1}\left( r_{j}-r_{j+1} \right).
			\end{align}
			The weakest Fisher weight is
			\begin{align}
				w(\bm{r}) =&~ \min\left\{ 2\sinh^{2}\left( 2r_{n} \right),4\sinh^{2}\delta \right\}.
			\end{align}
			In the chamber interior, no sum-root channel attains the minimum; on boundary strata a sum-root channel can tie at zero.
		\end{lemma}
		
		\begin{proof}
			For $i<j$,
			\begin{align}
				r_{i}-r_{j} =&~ \sum_{k=i}^{j-1}\left( r_{k}-r_{k+1} \right)\ge\delta,
			\end{align}
			so
			\begin{align}
				\min_{i<j}F_{ij}^{\mathrm{bsY}} =&~ 4\sinh^{2}\delta.
			\end{align}
			Ordering also gives
			\begin{align}
				r_{i}+r_{j} \ge&~ r_{n-1}+r_{n}\ge2r_{n},
			\end{align}
			placing every sum-root weight above the smallest phase-rotation weight,
			\begin{align}
				4\sinh^{2}\left( r_{i}+r_{j} \right) \ge&~ 4\sinh^{2}\left( 2r_{n} \right)\ge2\sinh^{2}\left( 2r_{n} \right).
			\end{align}
			Every sum-root weight is at least twice the weakest phase-rotation weight, and hence exceeds it whenever the latter is positive.
			Monotonicity of $\sinh^{2}\left( 2r \right)$ for $r\ge0$ places that weight at $r_{n}$.
		\end{proof}
		
		\begin{theorem}[Exact photon-number-constrained optimum]
			\label{thm:AP}
			For $n\ge2$ and fixed mean photon number
			\begin{align}
				N =&~ \sum_{j=1}^{n}\sinh^{2}r_{j} > 0,
			\end{align}
			the maximizer of $w(\bm{r})$ in the closed fundamental chamber $\overline{\mathcal{C}}$ is unique and is the arithmetic progression
			\begin{align}
				r_{j} =&~ a+(n-j)b,
				\qquad
				j=1,\ldots,n,
			\end{align}
			whose weakest phase-rotation and difference-root weights are balanced,
			\begin{align}
				2\sinh^{2}\left( 2a \right) =&~ 4\sinh^{2}b = w.
			\end{align}
			Equivalently, this spectrum uniquely minimizes the photon number required to attain any prescribed worst weight $w>0$.
			All other maximizers on the full flat are related to it by the restricted Weyl group, that is by mode permutations and squeezing-sign changes.
		\end{theorem}
		
		\begin{proof}
			Choose a target weight $w_{0}>0$.
			The reduced expression for $w(\bm{r})$ requires
			\begin{align}
				\begin{aligned}
					2\sinh^{2}\left( 2r_{n} \right) \ge&~ w_{0},
					\\
					4\sinh^{2}\left( r_{j}-r_{j+1} \right) \ge&~ w_{0},
					\qquad
					j=1,\ldots,n-1.
				\end{aligned}
			\end{align}
			Define $a$ and $b$ by
			\begin{align}
				2\sinh^{2}\left( 2a \right) =&~ 4\sinh^{2}b = w_{0}.
			\end{align}
			Monotonicity gives
			\begin{align}
				r_{n}\ge&~ a,
				\qquad
				r_{j}-r_{j+1}\ge b,
			\end{align}
			accumulating into
			\begin{align}
				r_{j}\ge&~ a+(n-j)b.
			\end{align}
			Because $\sum_{j}\sinh^{2}r_{j}$ increases strictly with every $r_{j}\ge0$, the minimum-resource spectrum saturates all bounds,
			\begin{align}
				r_{n} =&~ a,
				\qquad
				r_{j}-r_{j+1} = b.
			\end{align}
			These equalities give the arithmetic progression.
			The minimum photon number required to attain $w_{0}$ increases strictly with $w_{0}$, so inversion gives the same unique spectrum at fixed $N$.
		\end{proof}
		
		For $n=1$, no beam-splitter gap exists and $w = 2\sinh^{2}\left( 2r_{1} \right)$, so the single mode carries the full squeezing resource.
		
		%%%%%%%%%%%%%%%%%%%%%%%%%%%%%%%%%%%%%%%%%%%%%%%%%%%%%%%%%%%%%%%%%%%%%%%%%%%%%%
		\subsection{Invariant normalization and exact dual-Weyl optimum}
		\label{sec:invariantdesign}
		
		\begin{theorem}[Invariant E-optimal design]
			\label{thm:invariantdesign}
			For the invariantly orthonormal passive generators, all canonical Fisher weights are $\widetilde F_{\bm{\alpha}}=2\sinh^{2}[\bm{\alpha}(\bm{r})]$.
			At fixed $N>0$ their minimum is uniquely maximized in the closed fundamental chamber by the spectrum on the dual-Weyl ray with that photon number, for every $n\geq1$.
			The Fisher eigenvalues and this optimum are independent of the passive Bloch--Messiah frame.
		\end{theorem}
		
		\begin{proof}
			Dividing each beam-splitter generator by $\sqrt{2}$ divides its Fisher weight by two.
			For a target minimum root evaluation $t>0$, the simple-root constraints require $r_n\geq t/2$ and $r_j-r_{j+1}\geq t$.
			They imply every positive-root constraint and give $r_j\geq(n-j+1/2)t$.
			The photon number is strictly increasing in each nonnegative coordinate, so its unique minimum at fixed $t$ saturates every bound.
			The resulting minimum budget increases strictly with $t$, and inversion yields the unique optimum at fixed $N$, with $a=t/2$.
			For $n\times n$ Hermitian generator matrices $\bm{h}$ and $\bm{k}$ and a passive mode transformation $\bm{U}\in\mathrm{U}(n)$, conjugation $\bm{h}\mapsto\bm{U}^{\dagger}\bm{h}\bm{U}$ preserves their trace inner product $\operatorname{Re}\operatorname{Tr}(\bm{h}\bm{k})$~\cite{Arvind1995p471}.
			It acts by an orthogonal matrix on the normalized passive basis, and transforms $\widetilde{\bm{F}}$ by orthogonal similarity.
		\end{proof}
		
		For three modes at $N=2$, the angle-optimal spectrum gives an invariant minimum weight of $0.3016$.
		The invariant optimum has smallest squeezing $a\simeq0.2072$, exact ratios $5:3:1$, and $\widetilde w\simeq0.3637$.
		These are distinct objectives on the same passive model.
		
		%%%%%%%%%%%%%%%%%%%%%%%%%%%%%%%%%%%%%%%%%%%%%%%%%%%%%%%%%%%%%%%%%%%%%%%%%%%%%%
		\subsection{Weak- and strong-resource limits in angle coordinates}
		%%%%%%%%%%%%%%%%%%%%%%%%%%%%%%%%%%%%%%%%%%%%%%%%%%%%%%%%%%%%%%%%%%%%%%%%%%%%%%
		
		The balance between the weakest local-rotation and difference-root weights determines the adjacent gap as a function of the smallest squeezing,
		\begin{align}
			b(a) =&~
			\operatorname{arsinh}
			\left[
			\frac{1}{\sqrt{2}}
			\sinh\left( 2a \right)
			\right].
		\end{align}
		The corresponding optimal spectrum,
		\begin{align}
			r_{j}(a) =&~ a+(n-j)b(a),
			\qquad
			j=1,\ldots,n,
		\end{align}
		forms a one-parameter family with mean photon number
		\begin{align}
			N(a) =&~
			\sum_{j=1}^{n}
			\sinh^{2}
			\left[
			a+(n-j)b(a)
			\right].
		\end{align}
		Both $b(a)$ and every $r_{j}(a)$ increase strictly with $a\geq0$, so $N(a)$ also increases strictly.
		The limit $a\rightarrow0$ describes the weak-resource regime, while $a\rightarrow\infty$ describes the strong-resource regime.
		
		Since $\sinh b(a)=\sqrt{2}\sinh a\cosh a$ and $\sqrt{2}\cosh a>1$ for $a\ge0$, monotonicity of $\sinh$ gives $b(a)>a$ for every $a>0$.
		
		For weak squeezing,
		\begin{align}
			\sinh\left( 2a \right) =&~ 2a+\mathcal{O}\left( a^{3} \right),
		\end{align}
		and the balance relation gives
		\begin{align}
			b(a) =&~ \sqrt{2}\,a+\mathcal{O}\left( a^{3} \right).
		\end{align}
		The normalized optimal spectrum satisfies
		\begin{align}
			\frac{r_{j}(a)}{a} =&~
			1+\sqrt{2}(n-j)+\mathcal{O}\left( a^{2} \right),
			\qquad
			a\rightarrow0.
		\end{align}
		
		For strong squeezing,
		\begin{align}
			\sinh\left( 2a \right) =&~
			\frac{1}{2}e^{2a}
			\left[
			1+\mathcal{O}\left( e^{-4a} \right)
			\right],
		\end{align}
		and the large-argument expansion of the inverse hyperbolic sine gives
		\begin{align}
			b(a) =&~
			2a-\frac{1}{2}\ln 2
			+\mathcal{O}\left( e^{-4a} \right).
		\end{align}
		Substituting this expansion into the optimal arithmetic progression leaves
		\begin{align}
			r_{j}(a) =&~
			\left( 2n-2j+1 \right)a
			-\frac{n-j}{2}\ln 2
			+\mathcal{O}\left( e^{-4a} \right).
		\end{align}
		After normalization by the common diverging scale $a$,
		\begin{align}
			\frac{r_{j}(a)}{a} \longrightarrow&~
			2n-2j+1,
			\qquad
			j=1,\ldots,n,
			\qquad
			a\rightarrow\infty.
		\end{align}
		This limit identifies the dual-Weyl direction of the photon-number-constrained optimum.
		For finite photon number, the balance relation instead determines a resource-dependent ratio between the offset $a$ and the gap $b(a)$, so the exact optimum remains the arithmetic progression.

		\begin{proposition}[Condition number at the E-optimal spectrum]
			\label{prop:kappa}
			In the fundamental chamber and for $n\ge2$,
			\begin{align}
				\lambda_{\max}\left(\bm{F}_{\mathrm{passive}}\right) =&~
				\max\left\{ 2\sinh^{2}(2r_{1}),\, 4\sinh^{2}(r_{1}+r_{2}) \right\}.
			\end{align}
			For nonzero spectra, the condition number $\kappa=\lambda_{\max}/\lambda_{\min}$ is infinite on the singular arrangement and finite off it. At the vacuum, $\bm{F}=\bm{0}$ and this ratio is undefined.
			Along the optimal arithmetic progression,
			\begin{align}
				\ln\kappa\left[\bm{r}(a)\right] =&~ 4(n-1)\,b(a)+\mathcal{O}(1),
				\qquad a\rightarrow\infty,
			\end{align}
			so $\kappa$ grows without bound as the photon budget increases.
		\end{proposition}
		
		\begin{proof}
			Monotonicity of $\sinh^{2}$ on $[0,\infty)$ places the largest rotation argument at $2r_{1}$ and the largest sum-root argument at $r_{1}+r_{2}$.
			The largest difference-root argument is $r_{1}-r_{n}\le r_{1}+r_{2}$, and both carry the same coefficient $c=4$, so no $\mathrm{bsY}$ channel can exceed $\mathrm{bsX}_{12}$, leaving the largest rotation and sum-root weights as the two candidates for $\lambda_{\max}$.
			On $\mathcal{A}_{C_{n}}$ with $\bm{r}\neq\bm{0}$ one has $\lambda_{\min}=0$ while $\lambda_{\max}\ge2\sinh^{2}(2r_{1})>0$, giving $\kappa=\infty$; off the arrangement all weights are strictly positive and $\kappa$ is finite.
			For the asymptotics, $\sinh^{2}x=(1/4)e^{2x}[1+\mathcal{O}(e^{-2x})]$ and $b>(1/2)\ln 2$ at large $a$, so the maximum is realized by $\mathrm{rot}_{1}$ and $\lambda_{\max}=(1/2)e^{4r_{1}}[1+\mathcal{O}(e^{-4r_{1}})]$, while the balance condition gives $\lambda_{\min}=2\sinh^{2}(2a)=(1/2)e^{4a}[1+\mathcal{O}(e^{-4a})]$.
			Dividing and using $r_{1}-a=(n-1)b$ gives the condition-number asymptotics.
		\end{proof}
		
		\begin{remark}[Scope of the design criterion]
			E-optimality removes every Fisher-null direction for $N>0$. The condition-number growth established in Proposition~\ref{prop:kappa} is visible at $n=3$, where the optimal spectrum gives $\kappa\simeq25$ at $N=0.5$ and $\kappa\simeq3\times10^{3}$ at $N=50$.
			Growth with photon number does not establish whether this family minimizes the condition number at any fixed budget. That question remains open and requires a separate optimization argument.
		\end{remark}

		%%%%%%%%%%%%%%%%%%%%%%%%%%%%%%%%%%%%%%%%%%%%%%%%%%%%%%%%%%%%%%%%%%%%%%%%%%%%%%
		\section{Passive measurement compatibility}
		\label{sec:S8}
		%%%%%%%%%%%%%%%%%%%%%%%%%%%%%%%%%%%%%%%%%%%%%%%%%%%%%%%%%%%%%%%%%%%%%%%%%%%%%%
		
		The passive mean-commutator tensor is
		\begin{align}
			\mathcal{C}_{\mu\nu} =&~
			\frac{1}{2i}
			\left\langle
			\left[
			\hat{g}_{\mu},
			\hat{g}_{\nu}
			\right]
			\right\rangle.
		\end{align}
		For the canonical squeezed probe, the nonvanishing normal moments are
		\begin{align}
			\left\langle
			\hat{a}_{j}^{\dagger}\hat{a}_{k}
			\right\rangle =&~
			\delta_{jk}\sinh^{2}r_{j},
		\end{align}
		so every intermode first-order coherence vanishes.
		
		The local phase generators commute among themselves,
		\begin{align}
			\left[
			\hat{g}_{j}^{\mathrm{rot}},
			\hat{g}_{k}^{\mathrm{rot}}
			\right] =&~ 0.
		\end{align}
		Their commutators with beam-splitter generators are again intermode beam-splitter operators,
		\begin{align}
			\left[
			\hat{g}_{\ell}^{\mathrm{rot}},
			\hat{a}_{j}^{\dagger}\hat{a}_{k}
			\right] =&~
			\left(
			\delta_{\ell j}
			-
			\delta_{\ell k}
			\right)
			\hat{a}_{j}^{\dagger}\hat{a}_{k},
		\end{align}
		whose expectation values vanish for $j\neq k$,
		leaving
		\begin{align}
			\mathcal{C}_{\mathrm{rot}_{\ell},\mathrm{bsX}_{jk}} =&~
			\mathcal{C}_{\mathrm{rot}_{\ell},\mathrm{bsY}_{jk}}
			= 0.
		\end{align}
		
		Commutators between beam-splitter generators belonging to different unordered mode pairs either vanish or produce another intermode bilinear.
		Their expectation values vanish on the canonical product probe.
		The only nonzero passive mean commutator can occur between the two quadratures of the same mode pair.
		
		We write the intermode bilinear and its adjoint as
		\begin{align}
			\hat{A}_{jk} =&~
			\hat{a}_{j}^{\dagger}\hat{a}_{k},
			\qquad
			\hat{A}_{jk}^{\dagger}
			=
			\hat{a}_{k}^{\dagger}\hat{a}_{j}.
		\end{align}
		The beam-splitter quadratures are
		\begin{align}
			\begin{aligned}
				\hat{g}_{jk}^{\mathrm{bsX}} =&~
				\hat{A}_{jk}
				+
				\hat{A}_{jk}^{\dagger},
				\\
				\hat{g}_{jk}^{\mathrm{bsY}} =&~
				-i
				\left(
				\hat{A}_{jk}
				-
				\hat{A}_{jk}^{\dagger}
				\right).
			\end{aligned}
		\end{align}
		Using the bilinear commutator
		\begin{align}
			\left[
			\hat{A}_{jk},
			\hat{A}_{jk}^{\dagger}
			\right] =&~
			\hat{a}_{j}^{\dagger}\hat{a}_{j}
			-
			\hat{a}_{k}^{\dagger}\hat{a}_{k},
		\end{align}
		we obtain the beam-splitter commutator
		\begin{align}
			\left[
			\hat{g}_{jk}^{\mathrm{bsX}},
			\hat{g}_{jk}^{\mathrm{bsY}}
			\right] =&~
			2i
			\left(
			\hat{a}_{j}^{\dagger}\hat{a}_{j}
			-
			\hat{a}_{k}^{\dagger}\hat{a}_{k}
			\right).
		\end{align}
		Its expectation value gives the mean-commutator tensor entry
		\begin{align}
			\begin{aligned}
				\mathcal{C}_{\mathrm{bsX}_{jk},\mathrm{bsY}_{jk}} =&~
				\sinh^{2}r_{j}
				-
				\sinh^{2}r_{k}
				\\
				=&~
				\frac{1}{2}
				\left(
				\cosh 2r_{j}
				-
				\cosh 2r_{k}
				\right)
				\\
				=&~
				\sinh\left(r_{j}+r_{k}\right)
				\sinh\left(r_{j}-r_{k}\right).
			\end{aligned}
		\end{align}
		
		In the ordered passive basis
		\begin{align}
			\left\{
			\hat{g}_{j}^{\mathrm{rot}}
			\right\}_{j=1}^{n}
			\cup
			\left\{
			\hat{g}_{jk}^{\mathrm{bsX}},
			\hat{g}_{jk}^{\mathrm{bsY}}
			\right\}_{j<k},
		\end{align}
		the mean-commutator tensor has the direct-sum form
		\begin{align}
			\bm{\mathcal{C}}_{\mathrm{passive}} =&~
			\bm{0}_{n}
			\oplus
			\bigoplus_{1\leq j<k\leq n}
			\begin{pmatrix}
				0 &
				\mathcal{C}_{\mathrm{bsX}_{jk},\mathrm{bsY}_{jk}}
				\\
				-\mathcal{C}_{\mathrm{bsX}_{jk},\mathrm{bsY}_{jk}}
				&
				0
			\end{pmatrix}.
		\end{align}
		The passive Fisher weights,
		\begin{align}
			\begin{aligned}
				F_{jk}^{\mathrm{bsX}} =&~
				4\sinh^{2}\left(r_{j}+r_{k}\right),
				\\
				F_{jk}^{\mathrm{bsY}} =&~
				4\sinh^{2}\left(r_{j}-r_{k}\right),
			\end{aligned}
		\end{align}
		combine with the corresponding mean-commutator entry into the compact relation
		\begin{align}
			\left(
			\mathcal{C}_{\mathrm{bsX}_{jk},\mathrm{bsY}_{jk}}
			\right)^{2} =&~
			\frac{1}{16}
			F_{jk}^{\mathrm{bsX}}
			F_{jk}^{\mathrm{bsY}}.
		\end{align}
		Both beam-splitter quadratures are identifiable precisely when the right-hand side is positive, in which case their mean commutator is nonzero.
		Weak commutativity of the pair occurs only when one of its two restricted-root Fisher weights vanishes.
		
		\begin{proposition}[Saturated pair incompatibility and maximal compatible dimension]
			At any spectrum where both quadratures of a mode pair are identifiable, their $2\times2$ tangent Gram matrix has rank one and their dimensionless incompatibility measure is one.
			At a regular spectrum, the largest compatible real passive tangent subspace has dimension $n(n+1)/2$.
		\end{proposition}
		
		\begin{proof}
			For $|v_\mu\rangle=\Delta\hat g_\mu|\psi_{\bm{r}}\rangle$, the Gram matrix $Q_{\mu\nu}=\langle v_\mu|v_\nu\rangle=F_{\mu\nu}/4+i\mathcal C_{\mu\nu}$ is positive semidefinite.
			Its two-parameter determinant is nonnegative, giving the quantum-geometric incompatibility bound.
			For an identifiable beam-splitter pair, the off-diagonal Fisher entry vanishes, and the Fisher--commutator identity sets that determinant to zero while both diagonal entries are positive.
			The block has rank one, the two vectors are complex collinear, and the eigenvalues of $i\bm{F}^{-1/2}(4\bm{\mathcal{C}})\bm{F}^{-1/2}$ on the pair are $+1$ and $-1$.
			More generally $\bm{Q}\geq0$ and its complex conjugate imply that these normalized eigenvalues lie in $[-1,1]$.
			
			At a regular spectrum each mode pair contributes a nondegenerate two-dimensional antisymmetric block to $\bm{\mathcal{C}}$, while its kernel consists of the $n$ rotations.
			The quotient by that kernel is a symplectic vector space of dimension $n(n-1)$.
			An isotropic subspace in this quotient has dimension at most $n(n-1)/2$; restoring the kernel gives the bound $n(n+1)/2$.
			All rotations and one quadrature from every pair attain this bound.
		\end{proof}
		
		%%%%%%%%%%%%%%%%%%%%%%%%%%%%%%%%%%%%%%%%%%%%%%%%%%%%%%%%%%%%%%%%%%%%%%%%%%%%%%
		\section{Two-mode reduction and the Mach--Zehnder sloppy model}
		\label{sec:S8b}
		%%%%%%%%%%%%%%%%%%%%%%%%%%%%%%%%%%%%%%%%%%%%%%%%%%%%%%%%%%%%%%%%%%%%%%%%%%%%%%
		
		We compare the two-phase Mach--Zehnder model~\cite{Frigerio2026p2540001} with the difference-root wall of our classification to distinguish the two mechanisms.
		
		In that model, the estimated unitary is the ordered product
		\begin{align}
			\hat{U} =&~ \hat{U}_{R}(\lambda_{2})\,\hat{S}_{x}\,\hat{U}_{R}(\lambda_{1})\,\hat{U}_{\mathrm{BS}},
		\end{align}
		where $\hat{U}_{\mathrm{BS}}$ is a fixed parameter-independent beam splitter, $\hat{U}_{R}(\lambda)=e^{-i\lambda\hat{a}_{1}^{\dagger}\hat{a}_{1}}$ describes each of the two phase shifters on the same arm mode, and $\hat{S}_{x}$ is a single-mode squeezer of strength $x$, with fixed squeezing phase, inserted between them.
		The beam splitter precedes both phase encodings and carries no parameter, so it does not separate them.
		In the bare sloppy model, $x=0$, the two phase shifters act consecutively through the single generator $\hat{a}_{1}^{\dagger}\hat{a}_{1}$,
		\begin{align}
			e^{-i\lambda_{2}\hat{a}_{1}^{\dagger}\hat{a}_{1}}
			e^{-i\lambda_{1}\hat{a}_{1}^{\dagger}\hat{a}_{1}}
			=&~
			e^{-i(\lambda_{1}+\lambda_{2})\hat{a}_{1}^{\dagger}\hat{a}_{1}},
		\end{align}
		so the encoded state depends only on the sum $\lambda_{1}+\lambda_{2}$.
		The direction $(1,-1)$ is always Fisher-null; $(1,1)$ has a nonzero eigenvalue whenever the probe has nonzero QFI for the common phase, and otherwise both directions are null.
		This redundancy holds for every input squeezing spectrum and displacement, because both parameters use the same generator with nothing between them.
		
		The restricted-root singularities instead concern a single passive generator whose Fisher weight vanishes as the squeezing spectrum reaches a wall.
		For the difference root $\bm{e}_{1}-\bm{e}_{2}$, the relevant generator is $\hat{g}_{12}^{\mathrm{bsY}}$, and its weight $F_{12}^{\mathrm{bsY}}=4\sinh^{2}(r_{1}-r_{2})$ vanishes only on the wall $r_{1}=r_{2}$, becoming nonzero for unequal squeezing.
		This is a spectrum-controlled rank loss of one generator, not a spectrum-independent redundancy of a two-parameter encoding, and the two null directions are distinct objects.
		
		The two mechanisms require different remedies.
		An operation inserted between the two phase encodings can make their effective generators independent, as an intermediate squeezer does~\cite{Frigerio2026p2540001}.
		For the restricted-root wall, changing the squeezing spectrum makes the $\mathrm{bsY}$ tangent nonzero.
		The Mach--Zehnder input family allows relative displacement, but its compatible configuration also exists at zero displacement with nonzero input and intermediate squeezing~\cite{Frigerio2026p2540001}.
		Initial displacement leaves the bare sequential redundancy unchanged, while first moments can lift a covariance-wall blind direction.

		%%%%%%%%%%%%%%%%%%%%%%%%%%%%%%%%%%%%%%%%%%%%%%%%%%%%%%%%%%%%%%%%%%%%%%%%%%%%%%
		\section{Operational consequences of the singular arrangement}
		\label{sec:S9}
		%%%%%%%%%%%%%%%%%%%%%%%%%%%%%%%%%%%%%%%%%%%%%%%%%%%%%%%%%%%%%%%%%%%%%%%%%%%%%%
		
		%%%%%%%%%%%%%%%%%%%%%%%%%%%%%%%%%%%%%%%%%%%%%%%%%%%%%%%%%%%%%%%%%%%%%%%%%%%%%%
		\subsection{Phase sensing versus network calibration}
		%%%%%%%%%%%%%%%%%%%%%%%%%%%%%%%%%%%%%%%%%%%%%%%%%%%%%%%%%%%%%%%%%%%%%%%%%%%%%%
		
		For $n$ independent phase rotations, the canonical probe has Fisher matrix
		\begin{align}
			\bm{F}_{\mathrm{phase}} =&~ \mathrm{diag}\left( 2\sinh^{2}\left( 2r_{1} \right),\ldots,2\sinh^{2}\left( 2r_{n} \right) \right),
		\end{align}
		and total single-parameter inverse-QFI cost
		\begin{align}
			C_{\mathrm{phase}} =&~ \sum_{j=1}^{n}\frac{1}{2\sinh^{2}\left( 2r_{j} \right)}.
		\end{align}
		
		\begin{proposition}[Uniform phase-sensing optimum]
			\label{prop:uniform}
			At fixed mean photon number
			\begin{align}
				N =&~ \sum_{j=1}^{n}\sinh^{2}r_{j}>0,
			\end{align}
			uniform squeezing uniquely minimizes $C_{\mathrm{phase}}$ within the nonnegative squeezing chamber.
		\end{proposition}
		
		\begin{proof}
			With mean mode occupations
			\begin{align}
				n_{j} =&~ \sinh^{2}r_{j},
			\end{align}
			the identity $\sinh^{2}\left( 2r_{j} \right)=4n_{j}\left( n_{j}+1 \right)$ gives
			\begin{align}
				C_{\mathrm{phase}} =&~ \sum_{j=1}^{n}\frac{1}{8n_{j}\left( n_{j}+1 \right)}.
			\end{align}
			Any $n_{j}=0$ produces divergent cost.
			For the single-mode cost function
			\begin{align}
				f(x) =&~ \frac{1}{8x(x+1)},
			\end{align}
			strict convexity follows from
			\begin{align}
				f''(x) =&~ \frac{3x^{2}+3x+1}{4\left( x^{2}+x \right)^{3}}>0,
				\qquad
				x>0.
			\end{align}
			Jensen's inequality uniquely minimizes the sum at
			\begin{align}
				n_{1}=\cdots=n_{n} =&~ \frac{N}{n},
			\end{align}
			which corresponds to uniform squeezing.
		\end{proof}
		
		Passive network calibration probes both beam-splitter quadratures on each mode pair, with single-parameter inverse-QFI sensitivities
		\begin{align}
			\begin{aligned}
				\mathcal{S}_{ij}^{\mathrm{bsY}} =&~ \frac{1}{4\sinh^{2}\left( r_{i}-r_{j} \right)},
				\\
				\mathcal{S}_{ij}^{\mathrm{bsX}} =&~ \frac{1}{4\sinh^{2}\left( r_{i}+r_{j} \right)}.
			\end{aligned}
		\end{align}
		Uniform squeezing sets
		\begin{align}
			r_{i}-r_{j} =&~ 0
		\end{align}
		for every mode pair and makes every $\mathcal{S}_{ij}^{\mathrm{bsY}}$ diverge.
		The inverse-QFI costs characterize individual parameter sensitivities.
		Simultaneous attainability additionally requires weak commutativity and Fisher independence.
		
		For $n=3$ and $N=2$, uniform squeezing gives $\bm{r}_{\mathrm{phase}} = \left( 0.745,0.745,0.745 \right)$ and $C_{\mathrm{phase}} \simeq 0.34$, while every $\mathrm{bsY}$ sensitivity diverges.
		
		%%%%%%%%%%%%%%%%%%%%%%%%%%%%%%%%%%%%%%%%%%%%%%%%%%%%%%%%%%%%%%%%%%%%%%%%%%%%%%
		\subsection{Displacement lifting of a covariance-wall degeneracy}
		%%%%%%%%%%%%%%%%%%%%%%%%%%%%%%%%%%%%%%%%%%%%%%%%%%%%%%%%%%%%%%%%%%%%%%%%%%%%%%
		
		For an infinitesimal passive parameter change $\epsilon$ generated by $\bm{K}$,
		\begin{align}
			\delta\bm{d} =&~ \epsilon\bm{K}\bm{d},
			\qquad
			\delta\bm{\sigma} = \epsilon\left( \bm{K}\bm{\sigma}+\bm{\sigma}\bm{K}^{\mathrm{T}} \right).
		\end{align}
		On an equal-squeezing wall the corresponding $\mathrm{bsY}$ generator stabilizes the covariance, $\bm{K}\bm{\sigma}+\bm{\sigma}\bm{K}^{\mathrm{T}}=\bm{0}$, so the centered covariance model is Fisher-blind to that direction.
		
		\begin{proposition}[Displacement lifting]
			\label{prop:displacement}
			For a covariance-blind generator with $\bm{K}\bm{d}\neq\bm{0}$, the first moments of a displaced probe contribute a strictly positive Gaussian QFI~\cite{Monras2013p1303.3682,Safranek2019p035304},
			\begin{align}
				F_{\bm{d}} =&~ \left( \bm{K}\bm{d} \right)^{\mathrm{T}}\bm{\sigma}^{-1}\left( \bm{K}\bm{d} \right)>0.
			\end{align}
		\end{proposition}
		
		The covariance-stabilizer condition places the $\mathrm{bsY}$ generator in the isotropy subalgebra of $\bm{\sigma}$ on the wall, while the same generator acts nontrivially on the first moments whenever $\bm{K}\bm{d}\neq\bm{0}$.
		A coherent displacement recovers local sensitivity to a direction that the centered model cannot resolve, without moving the covariance off the wall.
		The magnitude of $F_{\bm{d}}$ scales with the squared displacement and draws on the displacement energy allocated to the probe.
		A displacement $\bm{d}$ adds $\lVert\bm{d}\rVert^{2}/2$ to the mean photon number $N=\sum_{j}\sinh^{2}r_{j}$, so it competes with squeezing for the same resource.
		Placing displacement on an equal footing with squeezing requires a joint optimization over the squeezing spectrum and the displacement at fixed total photon number. Our displacement result establishes that a covariance-wall degeneracy can be lifted; the optimal resource allocation remains open.
		
		%%%%%%%%%%%%%%%%%%%%%%%%%%%%%%%%%%%%%%%%%%%%%%%%%%%%%%%%%%%%%%%%%%%%%%%%%%%%%%
		\subsection{Robustness of the walls under isotropic additive noise}
		%%%%%%%%%%%%%%%%%%%%%%%%%%%%%%%%%%%%%%%%%%%%%%%%%%%%%%%%%%%%%%%%%%%%%%%%%%%%%%
		
		Isotropic additive Gaussian noise preserves the first moments and broadens the covariance~\cite{Weedbrook2012p621},
		\begin{align}
			\bm{\sigma}_{\bar{n}} =&~ \bm{\sigma}+\bar{n}\bm{I}_{2n},
			\qquad
			\bar{n}\ge0,
		\end{align}
		with $\bar{n}$ the added mean photon number per mode.
		
		\begin{proposition}[Isotropic additive-noise robustness]
			\label{prop:thermal}
			For the centered Gaussian family with covariance $\bm{\sigma}_{\bar n}=\bm{\sigma}+\bar n\bm{I}$ at any finite $\bar n\geq0$, the passive Fisher kernel is unchanged.
			Its QFI rank-loss arrangement is exactly $\mathcal A_{C_n}$, with the same nullity on every stratum.
		\end{proposition}
		
		\begin{proof}
			A passive generator commutes with the symplectic form, $\left[ \bm{\Omega},\bm{G} \right]=\bm{0}$, so its Hamiltonian matrix $\bm{K}=\bm{\Omega}\bm{G}$ is antisymmetric, $\bm{K}+\bm{K}^{\mathrm{T}}=\bm{0}$.
			Adding the isotropic term $\bar{n}\bm{I}$ to the covariance shifts the stabilizer condition by
			\begin{align}
				\bm{K}\left( \bar{n}\bm{I} \right)+\left( \bar{n}\bm{I} \right)\bm{K}^{\mathrm{T}} =&~ \bar{n}\left( \bm{K}+\bm{K}^{\mathrm{T}} \right) = \bm{0},
			\end{align}
			so $\bm{K}\bm{\sigma}_{\bar{n}}+\bm{\sigma}_{\bar{n}}\bm{K}^{\mathrm{T}}=\bm{K}\bm{\sigma}+\bm{\sigma}\bm{K}^{\mathrm{T}}$.
			A direction stabilizes $\bm{\sigma}_{\bar n}$ if and only if it stabilizes $\bm{\sigma}$.
			For $\bar n>0$ the Gaussian state is faithful because all symplectic eigenvalues are strictly larger than $1/2$.
			For its density operator $\hat{\rho}$, let $\lvert a\rangle$ be orthonormal eigenvectors with eigenvalues $p_{a}>0$.
			The derivative $\partial\hat{\rho}$ along a passive parameter direction has SLD QFI~\cite{Safranek2019p035304}
			\begin{align}
				F(\partial\hat{\rho}) =&~ 2\sum_{a,b}
				\frac{|\langle a|\partial\hat{\rho}|b\rangle|^2}{p_a+p_b},
			\end{align}
			which vanishes exactly when $\partial\hat{\rho}=0$.
			For a centered Gaussian state, the covariance determines the characteristic function, whose derivative vanishes exactly when $\partial\bm{\sigma}=0$.
			The QFI kernel equals the unchanged covariance-stabilizer kernel.
			The case $\bar n=0$ is Theorem~\ref{thm:passive-identifiability}.
		\end{proof}
		
		Isotropic additive Gaussian noise preserves the singular arrangement, while its nonzero Fisher weights require the mixed-state Gaussian QFI~\cite{Safranek2019p035304}.
		
		%%%%%%%%%%%%%%%%%%%%%%%%%%%%%%%%%%%%%%%%%%%%%%%%%%%%%%%%%%%%%%%%%%%%%%%%%%%%%%
		\section{Comparison of ordinary and restricted roots}
		\label{sec:S10}
		%%%%%%%%%%%%%%%%%%%%%%%%%%%%%%%%%%%%%%%%%%%%%%%%%%%%%%%%%%%%%%%%%%%%%%%%%%%%%%
		
		The ordinary and restricted $C_{n}$ constructions differ in their defining Abelian subspaces, adjoint actions, physical assignments, and metrological roles.
		
		\begin{table}[!ht]
			\caption{Ordinary and restricted root constructions of abstract type $C_{n}$.}
			\label{tab:rootcomparison}
			
			\centering
			
			\begin{tabular}{@{}lll@{}}
				\hline\hline
				
				\parbox[t]{0.18\textwidth}{}
				&
				\parbox[t]{0.36\textwidth}{\textbf{Ordinary Cartan--Weyl roots}}
				&
				\parbox[t]{0.36\textwidth}{\textbf{Restricted roots}}
				\\
				\hline
				
				\parbox[t]{0.18\textwidth}{Defining structure}
				&
				\parbox[t]{0.36\textwidth}{
					complexified algebra
					$\mathrm{sp}(2n,\mathbb{C})$
				}
				&
				\parbox[t]{0.36\textwidth}{
					symmetric pair
					$\mathrm{Sp}(2n,\mathbb{R})/\mathrm{U}(n)$
				}
				\\[2mm]
				
				\parbox[t]{0.18\textwidth}{Abelian subspace}
				&
				\parbox[t]{0.36\textwidth}{
					number Cartan
					$\mathfrak{h}
					=
					\mathrm{span}
					\left\{
					\hat{H}_{i}
					\right\}_{i=1}^{n}$
				}
				&
				\parbox[t]{0.36\textwidth}{
					diagonal squeezing flat
					$\mathfrak{a}
					=
					\mathrm{span}_{\mathbb{R}}
					\left\{
					\bm{A}_{i}
					\right\}_{i=1}^{n}
					\subset\mathfrak{p}$
				}
				\\[2mm]
				
				\parbox[t]{0.18\textwidth}{Coordinates}
				&
				\parbox[t]{0.36\textwidth}{
					number-Cartan dual basis
					$\bm{\varepsilon}_{i}(\hat{H}_{j})=\delta_{ij}$
				}
				&
				\parbox[t]{0.36\textwidth}{
					squeezing-flat dual basis
					$\bm{e}_{i}(\bm{r})=r_{i}$
				}
				\\[2mm]
				
				\parbox[t]{0.18\textwidth}{Positive roots}
				&
				\parbox[t]{0.36\textwidth}{
					$2\bm{\varepsilon}_{i}$,
					$\bm{\varepsilon}_{i}-\bm{\varepsilon}_{j}$,
					$\bm{\varepsilon}_{i}+\bm{\varepsilon}_{j}$
				}
				&
				\parbox[t]{0.36\textwidth}{
					$2\bm{e}_{i}$,
					$\bm{e}_{i}-\bm{e}_{j}$,
					$\bm{e}_{i}+\bm{e}_{j}$
				}
				\\[2mm]
				
				\parbox[t]{0.18\textwidth}{Characteristic action}
				&
				\parbox[t]{0.36\textwidth}{
					$\left[
					\hat{H},
					\hat{E}_{\bm{\alpha}}
					\right]
					=
					\bm{\alpha}(\hat{H})
					\hat{E}_{\bm{\alpha}}$,
					$\hat{H}\in\mathfrak{h}$
				}
				&
				\parbox[t]{0.36\textwidth}{
					$\left[
					\bm{A}(\bm{r}),
					\bm{K}_{\bm{\alpha}}
					\right]
					=
					\bm{\alpha}(\bm{r})
					\bm{P}_{\bm{\alpha}}$
				}
				\\[2mm]
				
				\parbox[t]{0.18\textwidth}{Root addition}
				&
				\parbox[t]{0.36\textwidth}{
					$\left[
					\hat{E}_{\bm{\alpha}},
					\hat{E}_{\bm{\beta}}
					\right]
					\propto
					\hat{E}_{\bm{\alpha}+\bm{\beta}}$
					when
					$\bm{\alpha}+\bm{\beta}$
					is a root
				}
				&
				\parbox[t]{0.36\textwidth}{
					not the organizing operation for passive Fisher weights
				}
				\\[2mm]
				
				\parbox[t]{0.18\textwidth}{Physical role}
				&
				\parbox[t]{0.36\textwidth}{
					quadratic-generator sectors,
					commutators,
					measurement compatibility
				}
				&
				\parbox[t]{0.36\textwidth}{
					radial response,
					passive Fisher weights,
					identifiability
				}
				\\[2mm]
				
				\parbox[t]{0.18\textwidth}{Beam-splitter assignment}
				&
				\parbox[t]{0.36\textwidth}{
					$\hat{g}_{ij}^{\mathrm{bsX}}$
					and
					$\hat{g}_{ij}^{\mathrm{bsY}}$
					belong to the ordinary sector
					$\pm\left(
					\bm{\varepsilon}_{i}-\bm{\varepsilon}_{j}
					\right)$
				}
				&
				\parbox[t]{0.36\textwidth}{
					$\hat{g}_{ij}^{\mathrm{bsY}}$
					corresponds to
					$\bm{e}_{i}-\bm{e}_{j}$,
					while
					$\hat{g}_{ij}^{\mathrm{bsX}}$
					corresponds to
					$\bm{e}_{i}+\bm{e}_{j}$
				}
				\\[2mm]
				
				\parbox[t]{0.18\textwidth}{Metrological tensor}
				&
				\parbox[t]{0.36\textwidth}{
					mean-commutator tensor
					$\bm{\mathcal{C}}$
				}
				&
				\parbox[t]{0.36\textwidth}{
					passive Fisher tensor
					$\bm{F}_{\mathrm{passive}}$
				}
				\\[2mm]
				
				\parbox[t]{0.18\textwidth}{Metrological condition}
				&
				\parbox[t]{0.36\textwidth}{
					weak commutativity,
					$\mathcal{C}_{\mu\nu}
					=
					(1/(2i))
					\left\langle
					\left[
					\hat{g}_{\mu},
					\hat{g}_{\nu}
					\right]
					\right\rangle
					=
					0$
				}
				&
				\parbox[t]{0.36\textwidth}{
					Fisher rank loss,
					$\bm{\alpha}(\bm{r})=0$
				}
				\\[2mm]
				
				\parbox[t]{0.18\textwidth}{Weight or coupling}
				&
				\parbox[t]{0.36\textwidth}{
					determined by structure constants and expectation values of the commutator output
				}
				&
				\parbox[t]{0.36\textwidth}{
					$F_{\bm{\alpha}}(\bm{r})
					=
					c_{\bm{\alpha}}
					\sinh^{2}
					\left[
					\bm{\alpha}(\bm{r})
					\right]$
				}
				\\
				
				\hline\hline
			\end{tabular}
			
		\end{table}
		
		\clearpage
		
		%%%%%%%%%%%%%%%%%%%%%%% References %%%%%%%%%%%%%%%%%%%%%%%%%
		
	%	\putbib[references]
	

	\end{bibunit}
	
\end{document}